\documentclass[10pt,twocolumn]{amsart} 

\usepackage[twocolumn,textwidth=16.31cm,columnsep=.81cm]{geometry}
\usepackage{amssymb,verbatim,mathabx,appendix}

\usepackage[shortlabels]{enumitem}
    \setenumerate[0]{label={\rm (\roman*)}, leftmargin=*}

\usepackage[T1]{fontenc}
\usepackage[english]{babel}

\usepackage[textsize=footnotesize,textwidth=20ex,colorinlistoftodos]{todonotes}

\usepackage{hyperref}
\hypersetup{
    colorlinks=true,
    linkcolor=blue,
    filecolor=magenta, 
    citecolor=red,     
    urlcolor=cyan,
}

\usepackage[capitalize]{cleveref}
    \crefname{enumi}{}{}
    \Crefname{enumi}{Item}{Items}
    \crefname{equation}{}{}
    \Crefname{equation}{Equation}{Equations}

\newtheorem{proposition}{Proposition}[section]

\newtheorem{theorem}[proposition]{Theorem}

\theoremstyle{definition}

\newtheorem{example}[proposition]{Example}

\newtheorem{remark}[proposition]{Remark}

\newcommand{\one}{\mathbf{1}}

\begin{document}
\title{Dynamical Lie algebras generated by Pauli strings and quadratic spaces over $\mathbb{F}_2$}
\author{Hans Cuypers} 
\maketitle


\section{Introduction}
In quantum physics and quantum science dynamical Lie algebras (and dynamical Lie groups) are a basic mathematical structure
that provides information on the symmetries of the studied quantum system.
They are Lie subalgebras of the unitary Lie algebra and relate to the Hamiltonian of 
a quantum systems.
Recent work has shown the importance of these Lie algebras in various areas
within quantum science. 
For example, in quantum control theory  the dynamical Lie algebra being the full special unitary algebra, indicates that one has full control over the system, see for example \cite{schirmer}. For 
variational quantum algorithms and quantum machine learning methods, the dimension of the Lie
algebra is directly related to phenomena like
barren plateaus  and overparametrization. See \cite{barren,barren2,barren3}.

Often  dynamical Lie algebras are given by a generating set of Pauli strings. 
Pauli strings are tensor products of the  Pauli matrices

$$I=\begin{pmatrix}1&0\\0&1\end{pmatrix}, X=\begin{pmatrix}0&1\\1&0\end{pmatrix},$$
$$Y=\begin{pmatrix}0&-i\\i&0\end{pmatrix}, Z=\begin{pmatrix}1&0\\0&-1\end{pmatrix}.$$

The set of all Pauli strings which are the tensor product of $n$ Pauli matrices
is denoted by $\mathcal{P}_n$. 
These Pauli strings generate the Pauli group $\Pi_n$ which consists of all scalar multiples of
the elements in $\mathcal{P}_n$ where the scalar is in $\{1,-1,i,-i\}$.
We will call all elements of $\Pi_n$ Pauli strings.

The elements in $\mathcal{P}_n\cup i\mathcal{P}_n$ 
form a basis of $M_{2^n}(\mathbb{C})$ the real vector space of all complex
$2^n\times 2^n$ matrices.
The matrices in  $\mathcal{P}_n$ are unitary matrices, while those in $i\mathcal{P}_n$ are anti-unitary.
So, we find the matrices in $i\mathcal{P}_n$ to be a basis for the Lie algebra $\mathfrak{u}(2^n)$,
the \emph{real unitary Lie algebra}, of anti-unitary matrices in $M_{2^n}(\mathbb{C})$.

Here we define the  Lie bracket $[\cdot,\cdot]$ on $M_{2^n}(\mathbb{C})$ by
$$[A,B]=\frac{1}{2}(AB-BA)$$
for all $A,B\in M_{2^n}(\mathbb{C})$.
The factor $\frac{1}{2}$ has the advantage that for any two Pauli strings $p,q\in \Pi_n$
we find 

$$[p,q]=\frac{1}{2}(pq-qp)=\begin{cases}pq\ \text{if} \ pq=-qp\\ 0\ \text{if} \ pq=qp.\end{cases}$$

So either $[p,q]$ is again a Pauli string or it is $0$.

The elements in $i\mathcal{P}_n$ different from $i I_{2^n}$, where $I_{2^n}$ denotes the $2^n\times 2^n$ identity matrix,
form a basis for the \emph{special real unitary Lie algebra} $\mathfrak{su}(2^n)$.

Dynamical Lie algebras, i.e. Lie subalgebras of $\mathfrak{su}(2^n)$,   generated by Pauli strings have recently been 
studied intensively. They are also called \emph{Pauli Lie algebras} or \emph{Hamiltonian Lie algebras}.
For example, in  \cite{fullclass} a complete classification of the isomorphism types of Pauli Lie algebras has been obtained, while \cite{wiersema1,wiersema2,freeLie} study Lie algebras  generated by Pauli strings with 2-local
interaction.
The papers \cite{generators},\cite{optimal} and \cite{minimalparitybasis} focus on minimal sets of Pauli strings needed to generate $\mathfrak{su}(2^n)$, or a specific subalgebra.

In this paper we provide a uniform mathematical approach to most of the results presented in these papers.
Moreover, we present an algorithm that on input of a set of Pauli strings determines the isomorphism type of the 
dynamical Lie algebra generated by these Pauli's in  time $\mathcal{O}(\max(n,m)^3)$ where $m$ is the size of the generating set.

Our approach makes use of the following. The set of all Pauli strings (and their scalar multiples by $\pm 1,\pm i$)
is a group $\Pi_n$ of order $4\cdot 4^n$. This group is called the Pauli group. It contains the normal subgroup $\mathcal{Z}_2= \{\pm I_{2^n}\}$ such that
$\Pi_n/\mathcal{Z}_2$ can be identified with an $\mathbb{F}_2$ vector space $V$ of dimension $2^{2n+1}$.
The map $p\in \Pi_n\mapsto p^2$ induces a quadratic form $Q$ on this vector space, which maps the elements
of $i\mathcal{P}_n$ to non-isotropic vectors, i.e. vectors $v$ with $Q(v)=1$.
Triples $p,q,r$ of $i\mathcal{P}_n$ with $[p,q]=\pm r$ correspond to elliptic lines with respect
to $Q$, i.e. $2$-dimensional subspaces on which $Q$ takes the value $1$ at all nonzero vectors. 

This implies that we can study Lie subalgebras $\mathfrak{g}$ of $\mathfrak{su}(2^n)$
generated by Pauli strings
by  subspaces of the  geometry $NO(V,Q)$ of non-isotropic points and elliptic lines of the quadratic space $(V,Q)$.

A classification of these subspaces leads to a classification of the  Lie subalgebras $\mathfrak{g}$ of $\mathfrak{su}(2^n)$
generated by Pauli strings.
Moreover, given a set of non-isotropic vectors in $V$ we can use straightforward linear algebra
over the field $\mathbb{F}_2$ and some graph theory to identify the subspace  of $NO(V,Q)$ generated by these vectors.
Using the correspondence with Lie subalgebras as described above, this translates to  an
efficient algorithm that given a set of generating Pauli strings
in $\mathfrak{su}(2^n)$ determines the isomorphism type of the Lie algebra generated by this set.

In the next three sections we introduce quadratic spaces $(V,Q)$ over $\mathbb{F}_2$ and the geometries
$NO(V,Q)$ on non-isotropic points and elliptic lines
and establish the connection with subgeometries of $NO(V,Q)$
and Lie subalgebras of $\mathfrak{su}(2^n)$ generated by Pauli strings.

\cref{sect:Lie} presents the classification of Lie subalgebras of   $\mathfrak{su}(2^n)$ generated by Pauli strings, while \cref{sect:generators} focuses on the connection between the frustration graph on a set of Pauli strings and the Lie subalgebra generated by them.

This connection forms the basis for the algorithm presented in \cref{sect:alg} to determine
the isomorphism type of a Lie subalgebra generated by a given set of generating Pauli strings.

In the two final sections we discuss the relation of our approach to recent work on 
Dynamical Lie algebras generated by Pauli strings.

Our approach is based on \cite{cuypers_clifford} where we consider Lie algebras (as well as Jordan algebras) of quasi-Clifford algebras over a field $\mathbb{F}$ of characteristic $\neq 2$, which are generated by  a set of elements $X$, being the vertex set of an ordinary graph $\Gamma=(X,E)$, 
such that for all $x,y\in X$
we have an element $\lambda_x\in \mathbb{F}$
such that

$$\begin{array}{lll} 
   x^2&=\lambda_x\one &\\
  {[x,y]}&=0&\textrm{if } \{x,y\}\not\in \mathcal{E},\\
  {[x,[x,y]]}&=\lambda_x y&\textrm{if } \{x,y\}\in \mathcal{E}.\\
\end{array}$$

A set of $2n+1$ Pauli strings
generating the Pauli group does also generate the real algebra $M_{2n}(\mathbb{C})$ as
a (quasi-)Clifford algebra, and satisfies the above
multiplication with $\lambda_x\in \{\pm 1\}$.

Here $\mathbf{1}$ the identity element.

Details and proofs of various statements in the main part of this article can be found in the appendices.

\section{The Pauli group and a quadratic $\mathbb{F}_2$-space}
\label{sect:pauli}
Let $n\geq 1$ be an integer and consider the set $\mathcal{P}_n\cup i\mathcal{P}_n$
of Pauli strings.
As for any two such Pauli strings $p,q\in\mathcal{P}_n\cup i\mathcal{P}_n$  we have 
$$p^2=q^2=\pm \one\ \mathrm{and} \ pq=\pm qp,$$
where $\one$ denotes the $2^n\times 2^n$-identity matrix,
we find that the group generated by all  Pauli strings is a group $\Pi_n$ of order $4\cdot 4^n=4\cdot 2^{2n}$.
This group is called the \emph{Pauli group}.

The group $$\mathcal{Z}_4=\langle i \one\rangle$$
is normal in $\Pi_n$ and the quotient 
$$\Pi_n/\mathcal{Z}_4$$
is an elementary abelian $2$-group which we can identify with an $\mathbb{F}_2$-vector space of dimension $2^{2n}$.
Here $\mathbb{F}_2$ is the field with two elements $0$ and $1$.
This vector space is equipped with a symplectic form $\overline{f}$ defined by
$$\overline{f}(p\mathcal{Z}_4,q\mathcal{Z}_4)=
\begin{cases}
1\ \mathrm{if} \ pq=-qp\\
0\ \mathrm{if} \ pq=qp.\\
\end{cases}$$

This space  plays a crucial role in all the above investigations on Dynamical Lie algebras generated by Pauli strings.
See \cite{fullclass,wiersema1,wiersema2,generators,optimal}.

In this paper, however, we notice that also the subgroup $\mathcal{Z}_2=\{ \pm \one\}$
is normal in $\Pi_n$ and the quotient
$$\Pi_n/\mathcal{Z}_2$$
is an elementary abelian $2$-group which we can identify with an $\mathbb{F}_2$-vector space $V$ of dimension $2n+1$. Here addition in $V$ corresponds to group multiplication in $\Pi_n/\mathcal{Z}_2$. For each $p\in \Pi_n$ we denote by $v_p$ the coset $p\mathcal{Z}_2$ and identify it with a vector in $V$.
So for Pauli strings $p,q$ we have 
$$v_{pq}=v_p+v_q.$$

The map

$$Q:V\rightarrow \mathbb{F}_2=\{0,1\}$$
defined by 
$$Q(v_p)=\begin{cases}
1\ \mathrm{if} \ p^2=-\one\\
0\ \mathrm{if} \ p^2=\one\\
\end{cases}$$
is a quadratic
form with associated symplectic form 
$f$ defined by
$$f(v_p, v_q)=
\begin{cases}
1\ \mathrm{if} \ pq=-qp\\
0\ \mathrm{if} \ pq=qp.\\
\end{cases}$$


The subspace $\langle v_{i\one}\rangle$ is the radical of the form $f$ but is non-isotropic with respect to $Q$,
which means $Q(v_{i\one})=1$. 
So $(V,Q)$ is a non-degenerate quadratic space of dimension $2n+1$.
The form $\overline{f}$ is the form induced by $f$ on $\Pi_n/\mathcal{Z}_4\simeq V/\langle v_{i\one} \rangle  $.

\begin{example}\label{basis}
For $P\in \{I,X,Y,Z\}$ we denote by $P_i$ the Pauli string
$$I\otimes \cdots \otimes P\otimes \cdots \otimes I$$
with $P$ being the $i^{th}$ tensor. Then 
we can take the generators
$$X_j,Z_j,i\one,\ \mathrm{where}\ 1\leq j\leq n$$
for the Pauli group and as a basis for $V$ the $2n+1$ vectors
$$\begin{array}{l}
x_1=v_{X_1},\dots,x_n= v_{X_n},\\
z_1=v_{Z_1}, \cdots, z_n=v_{Z_n},\\
r=v_{i\one}.\\
\end{array}$$

For $v\in V$ with coordinates $(v_1,\dots, v_{2n+1})$ with respect to this basis, we find
$$Q(v)=v_1v_{n+1}+\dots +v_{n}v_{2n}+v_{2n+1}^2.$$ 

Moreover, with $w=(w_1,\dots, w_{2n+1})\in V$ we find 
$$f(v,w)=v_1w_{n+1}+w_1v_{n+1}+\cdots+v_{n}w_{2n}+w_{n}v_{2n}.$$
\end{example}

We will study Lie subalgebras of $\mathfrak{su}(2^n)$ generated by Pauli strings
by relating them to certain subspaces and subgeometries of this quadratic space $(V,Q)$.

\section{Lie subalgebras and the geometry of elliptic lines}
\label{sect:Liegeometry}

We continue with the setting of the previous section.
The elements of $i\mathcal{P}_n$ map under the quotient map
$$\Pi_n\rightarrow V=\Pi_n/\mathcal{Z}_2$$ 
bijectively to the set of non-isotropic points of $V$, i.e., the 
elements $v\in V$ with $Q(v)=1$.

Two elements $p,q\in i\mathcal{P}_n$
satisfy
$$[p,q]=0\Leftrightarrow f(v_p,v_q)=0.$$
Moreover, 
$$[p,q]\neq 0\Leftrightarrow [p,q]\in \pm i\mathcal{P}_n\Leftrightarrow Q(v_p+v_q)=1,$$
and, if $[p,q]\neq 0$ then
$$v_{[p,q]}=v_p+v_q.$$

A $2$-space of $V$ is called \emph{elliptic} or of \emph{$-$-type} if for all three vectors $v\in V$ different from $0$
we have that $v$ is non-isotropic, so $Q(v)=1$.

In general, suppose $V$ is an $\mathbb{F}_2$ vector space equipped with 
a quadratic form $Q$ and associated symplectic form $f$.
By $NO(V,Q)$  (Non-isotropic points in the Orthogonal geometry of $(V,Q)$) we denote the \emph{point-line geometry} $(P,L)$
whose \emph{points} (elements of $P$) are the non-isotropic elements $v\in V_0$, which are not in the radical of $f_0$, 
and whose \emph{lines} (elements in $L$) 
are the triples of non-isotropic vectors in the elliptic 2-spaces
of $V$. 
Then in $NO(V,Q)$ any two points are contained in at most one line.

A \emph{subspace} of $NO(V,Q)$ is a subset $S$ of points which is closed under lines, meaning that for any two points $p,q$ in $S$
which are collinear we find the third point on the  line though $p$ and $q$ also in $S$.
As the intersection of two subspaces is again a subspace, we can define the \emph{subspace generated by} a subset $P_0$ to be the intersection
of all subspaces containing $P_0$. This subspace is denoted by $$\langle P_0\rangle_{NO}.$$
By $V_{P_0}$ we denote the linear subspace of $V$ generated by the elements of $P_0$.
Notice that $$\langle P_0\rangle_{NO}\subseteq V_{P_0}\cap P.$$

The \emph{collinearity graph} of $NO(V,Q)$ is the graph with $P$ as vertex set, and two vertices adjacent if and only if there is a line containing
both of them. We call a subset $P_0$ of $P$ \emph{connected} if the induced subgraph on the points in $P_0$ of the  collinearity graph is connected.

Now consider the Pauli group $Pi_n$ and identify the quotient $\Pi/\mathcal{Z}_2$ with 
the $\mathbb{F}_2$-vector space $V$ equipped with the quadratic form $Q$ and associated 
symplectic form $f$ as in \cref{sect:pauli}.

To a Lie subalgebra $\mathfrak{g}$ of $\mathfrak{su}(2^n)$ we can associate  a subset $$S_\mathfrak{g}=\{v_p\mid p\in i\mathcal{P}_n\setminus \{i\one\}\ \text{with}\ p\in \mathfrak{g}\}$$ of $NO(V,Q)$ and to a subset $S$ of $NO(V,Q)$ a Lie subalgebra
$$\mathfrak{g}_S=\langle p\in i\mathcal{P}_n| v_p\in S\rangle_{[\cdot,\cdot]},$$
the smallest linear subspace of $\mathfrak{g}$ closed under the Lie bracket $[\cdot,\cdot]$.

The following result provides the correspondence between Lie subalgebras of $\mathfrak{su}(2^n)$ generated by Pauli strings and subspaces of $NO(V,Q)$. 

\begin{proposition}
\phantom{a}
\begin{enumerate}
\item 
If $\mathfrak{g}$ is a Lie subalgebra of $\mathfrak{su}(2^n)$, then  $S_\mathfrak{g}$  is a subspace of $NO(V,Q)$.
\item 
If $S$ is a subspace of $NO(V,Q)$, then
$\mathfrak{g}_S$ is a Lie subalgebra of $\mathfrak{su}(2^n)$
 with $S_{\mathfrak{g}_S}=S.$
\end{enumerate}
\end{proposition}

\section{Subspaces of the geometry of elliptic lines}
\label{sect:subspaces}

Suppose $V$ is an $\mathbb{F}_2$ vector space equipped with 
a quadratic form $Q$ and associated symplectic form $f$.

The point-line geometry $NO(V,Q)$ has been extensively studied in the 70ties and 80ties in relation to
(simple) groups.
The geometry $NO(V,Q)$ has the  property that any two intersecting lines
generate a subspace consisting of $6$ points and $4$ lines, called a \emph{dual affine plane}
or \emph{symplectic plane}.
See \cref{fig:plane}

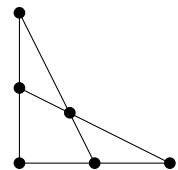
\begin{figure}
\begin{tikzpicture}
\filldraw[color=black]
(0,0) circle [radius=2pt]
(1,0) circle [radius=2pt]
(2,0) circle [radius=2pt]
(0.67,0.67) circle [radius=2pt]
(0,1) circle [radius=2pt]
(0,2) circle [radius=2pt];

\draw (0,0)--(2,0);
\draw (0,0)--(0,2);
\draw (2,0)--(0,1);
\draw (0,2)--(1,0);

\end{tikzpicture}
\caption{A dual affine plane}
\label{fig:plane}
\end{figure}
In particular, results of Jonathan Hall \cite{hall} imply that all connected subspaces of $NO(V,Q)$ are known.
There are two classes of such subspaces. The default class is obtained from linear subspaces of $V$.
Indeed if $W$ is a linear subspace of $V$ then an elliptic line containing two vectors from $W$ is contained in $W$, and the set of non-isotropic vectors of $W$ is a subspace of $NO(V,Q)$.
Notice that such subspace need not be connected, even if the 
quadratic form $Q_W$, being the restriction of $Q$ to $W$, is nondegenerate. It might happen that
the quadratic space $(W,Q_W)$ is of odd dimension with a non-isotropic vector in the radical of $f_W$, the restriction of $f$ to $W$,
which is an isolated vertex in the collinearity graph of the subspace, 
The  connected subspaces which we can obtain as $NO(W,Q_W)$ for some subspace
$W$ of $V_0$ will be called \emph{natural subspaces}.

The second class of connected subspaces can be described as follows.
Let $\Omega$ and $\Omega'$ be two disjoint finite sets.
By $T(\Omega,\Omega')$ (or $T(\Omega)$ in case $\Omega'=\emptyset$) we denote the partial linear space with
as point set all pairs $(T,T')$ of subsets $T$ of $\Omega$ of size $2$ and $T'$ of $\Omega'$ of
arbitrary size. A line of $T(\Omega,\Omega')$ is then a triple of points $\{x,y,z\}$
closed under taking the  symmetric difference.
So, $x=(T_x,T_x')$  and $y=(T_y,T_y')$ are collinear if and only if $|T_x\cap T_y|=1$.
The third point $z$ on the line through collinear $x$ and $y$ is then $z=(T_z,T_z')$
with $T_z=(T_x\cup T_y)\setminus (T_x\cap T_y)$ and $T_z'=(T_x'\cup T_y')\setminus (T_x'\cap T_y')$.

Now let $\hat{V}$ be the $\mathbb{F}_2$ space with basis $\Omega\cup \Omega'$.
We can identify any subset of $\Omega\cup \Omega'$ with the corresponding vector of $\hat{V}$
having precisely the coordinates in the subset equal to $1$.
Then let $V$ be  the subspace of $\hat{V}$ of all vectors with an even $\Omega$-weight.
By $P_0$ we denote the set of the elements of $V$ with $\Omega$-weight equal to $2$
and $\Omega'$-weight $0$.
On $V$ there is a unique quadratic form $Q$ with $Q(p)=1$ for all $p\in P_0$  and $Q(p)=0$ for all $p\in \mathbb{F}_2\Omega'$,
and as associated symplectic form $f$ which is the ordinary dot product on $\mathbb{F}_2\Omega$ and has $\Omega'$ in its radical.
We can identify the elements $T(\Omega,\Omega')$
with the corresponding non-isotropic vectors in $V$, lines corresponding to elliptic lines.  
This embedding of $T(\Omega,\Omega')$ inside $NO(V,Q)$ is called a \emph{standard embedding}
of $T(\Omega,\Omega')$. 

Notice that the radical of $f$ equals $\mathbb{F}_2\Omega'$,
if the size of $\Omega$ is odd, and 
$\mathbb{F}_2\Omega'\oplus \langle v_0\rangle$, where $v_0$ is the all-one vector of $\mathbb{F}_2\Omega$, if the size of $\Omega$ is even.
If $|\Omega|$ is even, but not a multiple of $4$, we find $Q(v_0)=1$, while 
if $4\mid|\Omega|$, then we find the vector $v_0$ to be isotropic 
and in the radical of $Q$.
This implies that when $|\Omega|$ is a multiple of $4$ we also have an embedding of $T(\Omega,\Omega')$ inside the space $NO(V/\langle v_0\rangle,\overline{Q})$, where $\overline{Q}$ is the form induced by $Q$ on $V/\langle v_0\rangle$.
We call this embedding the \emph{quotient of the standard embedding}.

From \cite{hall} we deduce the following classification of subspaces of $NO(V,Q)$:

\begin{theorem}
Let $S$ be a connected subspace of $NO(V,Q)$, then 
we have one of the following.

\begin{enumerate}
\item There is a subspace $W$ of $V$ such that $S$ equals $NO(W,Q_W)$;
\item There are disjoint sets $\Omega$ and $\Omega'$  such that 
$S$ is isomorphic to a standard  embedding of $T(\Omega,\Omega')$ in $V_S$, or, if $4\mid |\Omega|$ to its quotient.
\end{enumerate}
\end{theorem}

\section{Lie subalgebras and subspaces}
\label{sect:Lie}

As follows from the above sections, 
a Lie subalgebra $\mathfrak{g}$ of $\mathfrak{su}(2^n)$ generated by
Pauli strings corresponds to a subspace $S_\mathfrak{g}$ of $NO(V,Q)$,
where $V$ is isomorphic to $\Pi_n/\mathcal{Z}_2$ and  $Q$ the quadratic form as in \cref{sect:pauli} 

We often identify a subspace $S$ of $NO(V,Q)$ with the point-line geometry $(S,\{\ell\in L\mid \ell\subseteq S\})$.
This leads to the following result:

\begin{theorem}\label{samegeometry}
Two Lie subalgebras  $\mathfrak{g}$ and $\mathfrak{g}'$ of $\mathfrak{su}(2^n)$ generated by Pauli strings are isomorphic if
 $S_\mathfrak{g}$ and $S_{\mathfrak{g}'}$ are isomorphic.
\end{theorem}

\begin{remark}
It is not obvious that one can replace the 'if' in the above theorem by an 'if and only if'.
Two Lie subalgebras  $\mathfrak{g}$ and $\mathfrak{g}'$ of $\mathfrak{su}(2^n)$ generated by Pauli strings can be isomorphic, however,
the corresponding isomorphism need not map a Pauli string in $\mathfrak{g}$ to a 
Pauli string of $S_{\mathfrak{g}'}$.  
\end{remark}

Let $\mathfrak{g}$ be a Lie subalgebra of $\mathfrak{su}(2^n)$ generated by Pauli strings.
Denote by $\mathcal{P}_\mathfrak{g}$ the set of Pauli strings in $i\mathcal{P}_n$ contained in $\mathfrak{g}$.
For $p,q\in \mathcal{P}_\mathfrak{g}$ we define $p\equiv q$ if and only if for all Pauli strings in $\mathcal{P}_\mathfrak{g}$ 
 we have 
$$[p,r]=0\Leftrightarrow [q,r]=0.$$

The relation $\equiv$ is an equivalence relation.

If $\equiv$ is nontrivial, then we can find non-trivial ideals in $\mathfrak{g}$: 

\begin{proposition}\label{ideals}
Suppose $S_\mathfrak{g}$ is connected.
\begin{enumerate}
\item If $\equiv$ is trivial, then $\mathfrak{g}$ is simple.
\item Suppose $p_0\neq q_0\in \mathcal{P}_\mathfrak{g}$.
Then $$p_0\equiv q_0 \Leftrightarrow 0\neq v_{p_0}+v_{q_0}\in \mathrm{Rad}(Q_{V_\mathfrak{g}}).$$
\item 
If $p_0\neq q_0\in \mathcal{P}_\mathfrak{g}$ and $p_0\equiv q_0$ then with $z=p_0q_0$ we find
$$I_z^+=\langle (1+z)p\mid p\in \mathcal{P}_\mathfrak{g}\rangle$$ and 
$$I_{z}^-=\langle (1-z)p\mid p\in \mathcal{P}_\mathfrak{g}\rangle$$
to be ideals of $\mathfrak{g}$.

Moreover, $$\mathfrak{g}/I_z^{+}\simeq I_z^{+}\simeq I_z^{-}\simeq\mathfrak{g}/I_z^{-}$$
is isomorphic to the Lie subalgebra $\mathfrak{g}_{S_0}$, where $S_0$ is the intersection of $S$ with $W_0$, a hyperplane in the subspace of $V$ generated by $S_\mathfrak{g}$ not containing $v_{p_0}+v_{q_0}$.
\end{enumerate}
\end{proposition}

The above proposition shows that  
$\mathfrak{g}$ is isomorphic to a direct sum of $2^r$
copies of a simple Lie subalgebra, where $r$ is the dimension of the isotropic radical of
$Q_W$, where $W=V_{\mathfrak{g}}$ and $S_\mathfrak{g}$ is isomorphic to
$NO(W,Q_W)$ or the size of the set $\Omega'$ in case $S_\mathfrak{g}$ is isomorphic to $T(\Omega,\Omega')$ for disjoint sets $\Omega$ and $\Omega'$, where $|\Omega|\neq 4$.
(In  case $|\Omega|=4$ we find $T(\Omega,\Omega')$ to be isomorphic to $T(\Omega_1,\Omega_1')$
where $|\Omega_1|=3$ and $\Omega_1'|=|\Omega'|+1$.)

So, to identify the Lie subalgebras of $\mathfrak{su}(2^n)$ generated by Pauli strings
we need only consider the geometries $NO(W,Q_W)$ with $Q$ non-degenerate
and $T(\Omega,\Omega')$, where $|\Omega|\neq 4$ and $\Omega'$ is the empty set.

With the following three classes of Lie subalgebra of $\mathfrak{sl}(n)$,
the Lie algebra of all real traceless $n\times n$ matrices,
we can identify all Lie subalgebra of $\mathfrak{su}(2^n)$ generated by Pauli strings.

$$\begin{array}{ll}
\mathfrak{so}(n)&=\{x\in \mathfrak{sl}(n)\mid x=-x^\top \}\\
\mathfrak{su}(n)&=\{x\in \mathfrak{sl}(n)\mid x=-\overline{x}^\top \}\\
\mathfrak{sp}({n})&=\{x\in \mathfrak{su}(2n)\mid x=yx^\top y\}
\end{array}$$
where $$y=\begin{pmatrix} 0 & -1_{{n-1}}\\ 1_{{n-1}} &0\end{pmatrix}.$$

We obtain:

\begin{theorem}
Let $\mathfrak{g}$ be a Lie subalgebra of $\mathfrak{su}(2^n)$ generated by Pauli strings.
\begin{enumerate}
\item If $S_\mathfrak{g}$ is isomorphic to $NO(W,Q_W)$, for some non-degenerate quadratic space  $(W,Q_W)$
of odd dimension $2m+1$, then $\mathfrak{g}$ is isomorphic to $\mathfrak{su}(2^m)$;
\item If $S_\mathfrak{g}$ is isomorphic to $NO(W,Q_W)$, for some non-degenerate quadratic space $(W,Q_W)$
of even dimension $2m$ of $+$-type, then $\mathfrak{g}$ is isomorphic to $\mathfrak{so}(2^m)$;
\item If $S_\mathfrak{g}$ is isomorphic to $NO(W,Q_W)$, for some non-degenerate quadratic space $(W,Q_W)$
of even dimension $2m$ of $-$-type, then $\mathfrak{g}$ is isomorphic to $\mathfrak{sp}(2^m)$;
\item If $S_\mathfrak{g}$ is isomorphic to $T(\Omega,\emptyset)$ for some finite set $\Omega$ of size $3$ or at least $5$,
then $\mathfrak{g}$ is isomorphic to $\mathfrak{so}(|\Omega|)$.
\end{enumerate}

\end{theorem}

\begin{example}
Suppose $\Omega$ is a set of size $n$ with $n=3,5$ or $6$.
Then in the standard embedding of $T(\Omega)$ into the quadratic space $(V,Q)$ we find
$T(\Omega)$ to be equal to $NO(V,Q)$.
Indeed, we can identify $V$ with the hyperplane of $\mathbb{F}_2\Omega$
of vectors of even weight as described in \cref{sect:subspaces}.
The $\binom{n}{2}$ vectors of weight $2$
are non-isotropic. The $\binom{n}{4}$ vectors of weight $4$ are all  isotropic vectors.
If $n=6$, then the unique weight $6$ vector is anisotropic and in the radical of the quadratic form $Q$.

These isomorphisms of geometries imply the following isomorphism between Lie algebras:

$$\begin{array}{l}
\mathfrak{so}(3)\simeq \mathfrak{su}(2)\\
\mathfrak{so}(5)\simeq \mathfrak{sp}(2)\\
\mathfrak{so}(6)\simeq \mathfrak{su}(4)\\
\end{array}
$$

If $\Omega$ has size $8$, then we find  
the weight $8$ vector in $\mathbb{F}_2\Omega$ to span the  isotropic radical of $Q$.
So, in the quotient of the standard embedding we find that all 
anisotropic vectors are mapped to elements from $T(\Omega)$.
Hence if $\Omega$ has size $8$, then the quotient of the standard embedding maps $T(\Omega)$
isomorphically to $N(V,Q)$, where $(V,Q)$ is 6-dimensional of $+$-type.
We find two  representations of $\mathfrak{so}(8)=\mathfrak{so}(2^3)$ to be isomorphic.

If $\Omega$ has size $7$ or larger than $8$, we do find that weight $6$ vectors of $\mathbb{F}_2$
are anisotropic vectors in $V$ in both  the
standard embedding of $T(\Omega)$ or its quotient.
So, $T(\Omega)$ yields a proper subspace of $NO(V,Q)$.
\end{example}

The above two results provide the following  classification, also obtained in  \cite{fullclass}:

\begin{theorem}[Classification of Lie subalgebras generated by Pauli strings]\label{classification}
Let $\mathfrak{g}$ be a Lie subalgebra of $\mathfrak{su}(2^n)$ generated by Pauli strings.
Then $$\mathfrak{g}\simeq\bigoplus_{i=1}^k \mathfrak{g}_i,$$ 
with $\mathfrak{g}_i$ isomorphic to $\mathfrak{su}(2^{m_i})$,
$\mathfrak{so}(2^{m_i})$,  $\mathfrak{sp}(2^{m_i})$ or $\mathfrak{so}(m_i)$ for some $m_i$.
\end{theorem}

In the next section we provide a more detailed description of the Lie subalgebras of $\mathfrak{su}(2^n)$.

\section{Generators and frustration graphs}
\label{sect:generators}

Let $\mathcal{G}$ be a set of Pauli strings in $\mathfrak{su}(2^n)$.
Then the \emph{frustration graph} (or \emph{anti-commuting graph}) of $\mathcal{G}$ is the graph
with as vertices the elements from $\mathcal{G}$, two vertices $p,q\in \mathcal{G}$ adjacent if and only if $[p,q]\neq 0$
if and only if $f(v_p,v_q)\neq 0$. This graph we denote by $\Gamma_\mathcal{G}$.

The  Lie subalgebra $\mathfrak{g}$ of $\mathfrak{su}(2^n)$ generated by $\mathcal{G}$
equals $\mathfrak{g}_S$, where $S$ is the subspace of $NO(V,Q)$ generated by $\{v_p\mid p\in \mathcal{G}\}$, where $(V,Q)$ is the quadratic space as constructed in \cref{sect:pauli}.

This subspace is isomorphic to $T(\Omega,\Omega')$
for some disjoint set $\Omega$ and $\Omega'$ if and only if $\Gamma_\mathcal{G}$ is the line graph of a multi-graph $\Delta$ (see \cite{cuypersgraphs}).
Here a \emph{multi-graph} $\Delta$ is a graph with possibly multiple edges on two vertices, but no loops. Its \emph{line graph} is the (ordinary) graph with as vertices
the edges of $\Delta$ and two such vertices adjacent if and only if as edges in $\Delta$ they share a single vertex. 

We have the following characterization of frustration graphs which are line graphs of multi-graphs,
see \cite{cuypersgraphs} and also \cite{Chapman2020}:

\begin{theorem}\label{E6theorem}
Suppose $\mathcal{G}$ is a set of  Pauli strings in $\mathfrak{su}(2^n)$ with connected frustration graph $\Gamma_\mathcal{G}$.
Then the following three statements are equivalent:
\begin{enumerate}
\item $\Gamma_\mathcal{G}$ is the line graph of a multi-graph;
\item $\Gamma_\mathcal{G}$ does not contain one of the $32$ graphs
in \cref{E6graphs};
\item for no subset $\mathcal{G}_0$ with $6$ elements of $\mathcal{G}$ 
we have $\{v_p\mid p\in \mathcal{G}_0\}$ generates a subspace $NO(V,Q)$ isomorphic to $NO(W,Q_W)$ where $(W,Q_W)$ is a
non-degenerate $6$-dimensional space of $-$-type;
\item for no subset $\mathcal{G}_0$ with $6$ elements of $\mathcal{G}$ 
we have $\{p\mid p\in \mathcal{G}_0\}$ generates a Lie subalgebra isomorphic to $\mathfrak{sp}(4)$.
\end{enumerate}
\end{theorem}

\begin{remark}
The $32$ graphs in \cref{E6graphs} are precisely the frustration graphs of  sets  of $6$ generators $\mathcal{G}$
that in $V$ generate a non-degenerate subspace of $-$-type.
So, in $\mathfrak{g}$ these $6$ Pauli strings generate a Lie subalgebra isomorphic to $\mathfrak{sp}(2^2)$.
\end{remark}

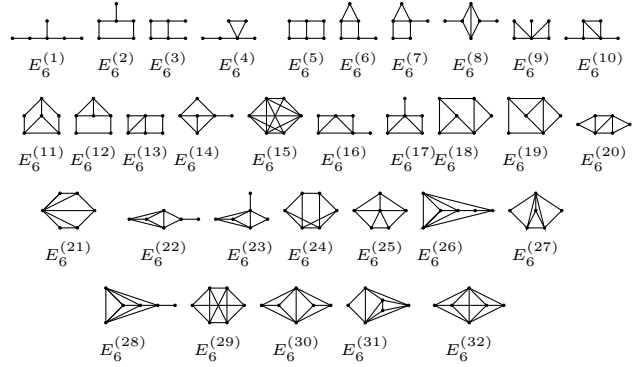
\begin{figure}
{\tiny{
\begin{tikzpicture}[scale=0.23]
\filldraw[color=black] 
(0,0) circle [radius=2pt]
(1,0) circle [radius=2pt]
(2,0) circle [radius=2pt]
(2,1) circle [radius=2pt]
(3,0) circle [radius=2pt]
(4,0) circle [radius=2pt];

\draw (0,0)--(4,0);
\draw (2,0)--(2,1);
\draw (2,-1.5) node {$E_6^{(1)}$};

\filldraw[color=black] 
(5,0) circle [radius=2pt]
(5,1) circle [radius=2pt]
(6,1) circle [radius=2pt]
(6,2) circle [radius=2pt]
(7,1) circle [radius=2pt]
(7,0) circle [radius=2pt];

\draw (7,0)--(5,0)--(5,1)--(7,1)--(7,0);
\draw (6,1)--(6,2);
\draw (6,-1.5) node {$E_6^{(2)}$};

\filldraw[color=black] 
(8,0) circle [radius=2pt]
(9,0) circle [radius=2pt]
(10,0) circle [radius=2pt]
(8,1) circle [radius=2pt]
(9,1) circle [radius=2pt]
(10,1) circle [radius=2pt];

\draw (8,0)--(10,0);
\draw (8,0)--(8,1)--(10,1);
\draw (9,0)--(9,1);
\draw (9,-1.5) node {$E_6^{(3)}$};

\filldraw[color=black] 
(11,0) circle [radius=2pt]
(12,0) circle [radius=2pt]
(13,0) circle [radius=2pt]
(14,0) circle [radius=2pt]
(12.5,1) circle [radius=2pt]
(13.5,1) circle [radius=2pt];

\draw (11,0)--(14,0);
\draw (13,0)--(12.5,1)--(13.5,1)--(13,0);
\draw (13,-1.5) node {$E_6^{(4)}$};

\filldraw[color=black] 
(16,0) circle [radius=2pt]
(17,0) circle [radius=2pt]
(18,0) circle [radius=2pt]
(16,1) circle [radius=2pt]
(17,1) circle [radius=2pt]
(18,1) circle [radius=2pt];

\draw (16,0)--(18,0)--(18,1)--(16,1)--(16,0);
\draw (17,0)--(17,1);
\draw (17,-1.5) node {$E_6^{(5)}$};

\filldraw[color=black] 
(19,0) circle [radius=2pt]
(20,0) circle [radius=2pt]
(19.5,2) circle [radius=2pt]
(20,1) circle [radius=2pt]
(21,0)  circle [radius=2pt]
(19,1) circle [radius=2pt];

\draw (19,0)--(20,0)--(20,1)--(19,1)--(19,0);
\draw (19,1)--(19.5,2)--(20,1);
\draw (20,0)--(21,0);
\draw (20,-1.5) node {$E_6^{(6)}$};

\filldraw[color=black] 
(22,0) circle [radius=2pt]
(23,0) circle [radius=2pt]
(22.5,2) circle [radius=2pt]
(23,1) circle [radius=2pt]
(24,1)  circle [radius=2pt]
(22,1) circle [radius=2pt];

\draw (22,0)--(23,0)--(23,1)--(22,1)--(22,0);
\draw (22,1)--(22.5,2)--(23,1);
\draw (23,1)--(24,1);
\draw (23,-1.5) node {$E_6^{(7)}$};

\filldraw[color=black] 
(25,1) circle [radius=2pt]
(26,1) circle [radius=2pt]
(26.5,2) circle [radius=2pt]
(26.5,0) circle [radius=2pt]
(27,1)  circle [radius=2pt]
(28,1) circle [radius=2pt];

\draw (25,1)--(26,1)--(26.5,2)--(26.5,0)--(27,1)--(28,1);
\draw (26,1)--(26.5,0);
\draw (26.5,2)--(27,1);
\draw (26.5,-1.5) node {$E_6^{(8)}$};

\filldraw[color=black] 
(29,0) circle [radius=2pt]
(29,1) circle [radius=2pt]
(30,0) circle [radius=2pt]
(31,0) circle [radius=2pt]
(31,1)  circle [radius=2pt]
(30,1) circle [radius=2pt];

\draw (29,0)--(31,0);
\draw (29,0)--(29,1)--(30,0)--(31,1)--(31,0);
\draw (30,0)--(30,1);
\draw (30,-1.5) node {$E_6^{(9)}$};

\filldraw[color=black] 
(32,0) circle [radius=2pt]
(33,0) circle [radius=2pt]
(34,0) circle [radius=2pt]
(35,0) circle [radius=2pt]
(33,1)  circle [radius=2pt]
(34,1) circle [radius=2pt];

\draw (32,0)--(35,0);
\draw (33,0)--(33,1)--(34,1)--(34,0)--(33,1);
\draw (34,-1.5) node {$E_6^{(10)}$};

\end{tikzpicture}

\medskip

\begin{tikzpicture}[scale=0.23]

\filldraw[color=black] 
(0,0) circle [radius=2pt]
(2,0) circle [radius=2pt]
(0,1) circle [radius=2pt]
(2,1) circle [radius=2pt]
(1,1)  circle [radius=2pt]
(1,2) circle [radius=2pt];

\draw (0,0)--(2,0)--(2,1)--(1,2)--(0,1)--(0,0);
\draw (0,0)--(1,1)--(2,0);
\draw (1,1)--(1,2);
\draw (1,-1.5) node {$E_6^{(11)}$};

\filldraw[color=black] 
(3,0) circle [radius=2pt]
(5,0) circle [radius=2pt]
(3,1) circle [radius=2pt]
(5,1) circle [radius=2pt]
(4,1)  circle [radius=2pt]
(4,2) circle [radius=2pt];

\draw (3,0)--(5,0)--(5,1)--(4,2)--(3,1)--(3,0);
\draw (4,1)--(4,2);
\draw (3,1)--(5,1);
\draw (4,-1.5) node {$E_6^{(12)}$};

\filldraw[color=black] 
(6,0) circle [radius=2pt]
(7,0) circle [radius=2pt]
(8,0) circle [radius=2pt]
(6,1) circle [radius=2pt]
(7,1)  circle [radius=2pt]
(8,1) circle [radius=2pt];

\draw (6,0)--(8,0)--(8,1)--(6,1)--(6,0)--(7,1)--(7,0);
\draw (7,-1.5) node {$E_6^{(13)}$};
\filldraw[color=black] 
(9,1) circle [radius=2pt]
(10,1) circle [radius=2pt]
(11,1) circle [radius=2pt]
(12,1) circle [radius=2pt]
(10,2)  circle [radius=2pt]
(10,0) circle [radius=2pt];

\draw (12,1)--(9,1)--(10,0)--(11,1);
\draw (10,0)--(10,1);
\draw (9,1)--(10,2)--(11,1);
\draw (10,-1.5) node {$E_6^{(14)}$};



\filldraw[color=black] 
(13,1) circle [radius=2pt]
(14,0) circle [radius=2pt]
(15,0) circle [radius=2pt]
(16,1) circle [radius=2pt]
(15,2)  circle [radius=2pt]
(14,2) circle [radius=2pt];

\draw (13,1)--(14,0)--(15,0)--(16,1)--(15,2)--(14,2)--(13,1);
\draw (13,1)--(16,1);
\draw (14,0)--(15,2);
\draw (15,0)--(14,2);
\draw (14,0)--(14,2);

\draw (13,1)--(15,0);
\draw (16,1)--(14,2);
\draw (14,0)--(16,1);

\draw (14.5,-1.5) node {$E_6^{(15)}$};

\filldraw[color=black] 
(17,1) circle [radius=2pt]
(18,1) circle [radius=2pt]
(19,1) circle [radius=2pt]
(17,0) circle [radius=2pt]
(19,0)  circle [radius=2pt]
(20,0) circle [radius=2pt];

\draw (17,1)--(19,1);
\draw (17,0)--(20,0);
\draw (17,0)--(17,1);
\draw (17,0)--(18,1);
\draw (19,0)--(18,1);
\draw (19,0)--(19,1);

\draw (18.5,-1.5) node {$E_6^{(16)}$};

\filldraw[color=black] 
(21,1) circle [radius=2pt]
(22,1) circle [radius=2pt]
(23,1) circle [radius=2pt]
(21,0) circle [radius=2pt]
(23,0)  circle [radius=2pt]
(22,2) circle [radius=2pt];

\draw (21,1)--(23,1);
\draw (21,0)--(23,0);
\draw (21,0)--(21,1);
\draw (21,0)--(22,1);
\draw (23,0)--(22,1);
\draw (22,1)--(22,2);
\draw (23,0)--(23,1);
\draw (22.5,-1.5) node {$E_6^{(17)}$};

\filldraw[color=black] 
(24,2) circle [radius=2pt]
(24,0) circle [radius=2pt]
(26,0) circle [radius=2pt]
(26,2) circle [radius=2pt]
(25,1)  circle [radius=2pt]
(27,1) circle [radius=2pt];

\draw (24,2)--(24,0)--(26,0)--(26,2)--(24,2);
\draw (24,2)--(26,0);
\draw (24,0)--(25,1);
\draw (26,2)--(27,1)--(26,0);

\draw (25,-1.5) node {$E_6^{(18)}$};

\filldraw[color=black] 
(28,2) circle [radius=2pt]
(28,0) circle [radius=2pt]
(30,0) circle [radius=2pt]
(30,2) circle [radius=2pt]
(29,1)  circle [radius=2pt]
(31,1) circle [radius=2pt];

\draw (28,2)--(28,0)--(30,0)--(30,2)--(28,2);
\draw (28,2)--(30,0);
\draw (30,2)--(29,1);
\draw (30,2)--(31,1)--(30,0);

\draw (29,-1.5) node {$E_6^{(19)}$};

\filldraw[color=black] 
(32,0.5) circle [radius=2pt]
(33,0) circle [radius=2pt]
(34,0) circle [radius=2pt]
(35,0.5) circle [radius=2pt]
(33,1)  circle [radius=2pt]
(34,1) circle [radius=2pt];

\draw (32,0.5)--(33,0)--(34,0)--(35,0.5)--(34,1)--(33,1)--(32,0.5);
\draw (33,0)--(34,1);
\draw (33,0)--(33,1);
\draw (34,0)--(34,1);
\draw (33.5,-1.5) node {$E_6^{(20)}$};

\end{tikzpicture}

\medskip

\begin{tikzpicture}[scale=0.23]
\filldraw[color=black] 
(-5,1) circle [radius=2pt]
(-4,0) circle [radius=2pt]
(-3,0) circle [radius=2pt]
(-2,1) circle [radius=2pt]
(-3,2)  circle [radius=2pt]
(-4,2) circle [radius=2pt];

\draw (-5,1)--(-4,0)--(-3,0)--(-2,1)--(-3,2)--(-4,2)--(-5,1);
\draw (-5,1)--(-3,0);
\draw (-5,1)--(-2,1);
\draw (-5,1)--(-3,2);
\draw (-5,1)--(-4,2);

\draw (-3.5,-1.5) node {$E_6^{(21)}$};

\filldraw[color=black] 
(0,0.5) circle [radius=2pt]
(1,0.5) circle [radius=2pt]
(2,0) circle [radius=2pt]
(2,1) circle [radius=2pt]
(3,0.5)  circle [radius=2pt]
(4,0.5) circle [radius=2pt];

\draw (0,0.5)--(1,0.5)--(2,0)--(2,1)--(1,0.5);
\draw (0,0.5)--(2,1)--(3,0.5)--(4,0.5);
\draw (0,0.5)--(2,0)--(3,0.5);

\draw (2,-1.5) node {$E_6^{(22)}$};

\filldraw[color=black] 
(5,0.5) circle [radius=2pt]
(6,0.5) circle [radius=2pt]
(7,0) circle [radius=2pt]
(7,1) circle [radius=2pt]
(8,0.5)  circle [radius=2pt]
(7,2) circle [radius=2pt];

\draw (5,0.5)--(6,0.5)--(7,0)--(7,1)--(6,0.5);
\draw (5,0.5)--(7,1)--(8,0.5);
\draw (7,1)--(7,2);
\draw (5,0.5)--(7,0)--(8,0.5);

\draw (7,-1.5) node {$E_6^{(23)}$};
\filldraw[color=black] 
(9,1) circle [radius=2pt]
(10,0) circle [radius=2pt]
(11,0) circle [radius=2pt]
(12,1) circle [radius=2pt]
(11,2)  circle [radius=2pt]
(10,2) circle [radius=2pt];

\draw (9,1)--(10,0)--(11,0)--(12,1)--(11,2)--(10,2)--(9,1);
\draw (10,0)--(10,2);
\draw (11,0)--(11,2);
\draw (9,1)--(11,0);
\draw (12,1)--(10,0);

\draw (10.5,-1.5) node {$E_6^{(24)}$};
\filldraw[color=black] 
(13,1) circle [radius=2pt]
(14,0) circle [radius=2pt]
(15,0) circle [radius=2pt]
(16,1) circle [radius=2pt]
(14.5,2)  circle [radius=2pt]
(14.5,1) circle [radius=2pt];

\draw (13,1)--(14,0)--(15,0)--(16,1)--(14.5,2)--(13,1)--(14.5,1)--(14,0);
\draw (15,0)--(14.5,1)--(16,1);
\draw (14.5,2)--(14.5,1);

\draw (14.5,-1.5) node {$E_6^{(25)}$};

\filldraw[color=black] 
(17,0) circle [radius=2pt]
(17,2) circle [radius=2pt]
(18,1) circle [radius=2pt]
(19,1) circle [radius=2pt]
(20,1)  circle [radius=2pt]
(21,1) circle [radius=2pt];

\draw (17,0)--(17,2)--(21,1)--(18,1)--(17,0)--(18,1)--(17,2);
\draw (21,1)--(17,0)--(19,1)--(17,2);

\draw (18,-1.5) node {$E_6^{(26)}$};

\filldraw[color=black] 
(22,1) circle [radius=2pt]
(23,0) circle [radius=2pt]
(24,0) circle [radius=2pt]
(25,1) circle [radius=2pt]
(23.5,2)  circle [radius=2pt]
(23.5,1) circle [radius=2pt];

\draw (22,1)--(23,0)--(24,0)--(25,1)--(23.5,2)--(22,1);

\draw (23.5,2)--(24,0)--(23.5,1);
\draw (23.5,2)--(23,0)--(23.5,1);

\draw (23.5,2)--(23.5,1);

\draw (23.5,-1.5) node {$E_6^{(27)}$};

\end{tikzpicture}

\medskip

\begin{tikzpicture}[scale=0.23]

\filldraw[color=black] 
(-5,0) circle [radius=2pt]
(-5,2) circle [radius=2pt]
(-4,1) circle [radius=2pt]
(-3,1) circle [radius=2pt]
(-2,1)  circle [radius=2pt]
(-1,1) circle [radius=2pt];

\draw (-5,0)--(-5,2)--(-4,1)--(-5,0);
\draw (-4,1)--(-1,1);
\draw (-5,0)--(-4,1)--(-5,2);
\draw (-5,0)--(-3,1)--(-5,2);
\draw (-5,0)--(-2,1)--(-5,2);

\draw (-4,-1.5) node {$E_6^{(28)}$};

\filldraw[color=black] 
(0,1) circle [radius=2pt]
(1,0) circle [radius=2pt]
(2,0) circle [radius=2pt]
(3,1) circle [radius=2pt]
(2,2)  circle [radius=2pt]
(1,2) circle [radius=2pt];

\draw (0,1)--(1,0)--(2,0)--(3,1)--(2,2)--(1,2)--(0,1);
\draw (1,0)--(1,2);
\draw (2,0)--(2,2);
\draw (1,2)--(2,0);
\draw (2,2)--(1,0);
\draw (0,1)--(3,1);
\draw (1.5,-1.5) node {$E_6^{(29)}$};

\filldraw[color=black] 
(4,1) circle [radius=2pt]
(5,1) circle [radius=2pt]
(6,2) circle [radius=2pt]
(6,0) circle [radius=2pt]
(7,1)  circle [radius=2pt]
(8,1) circle [radius=2pt];

\draw (4,1)--(6,2)--(8,1)--(6,0);
\draw (4,1)--(5,1)--(6,2)--(7,1)--(8,1);
\draw (5,1)--(6,0)--(7,1);
\draw (4,1)--(6,0)--(6,2);

\draw (6,-1.5) node {$E_6^{(30)}$};

\filldraw[color=black] 
(9,1) circle [radius=2pt]
(10,2) circle [radius=2pt]
(10,0) circle [radius=2pt]
(11,0.7) circle [radius=2pt]
(11,1.3)  circle [radius=2pt]
(12.5,1) circle [radius=2pt];

\draw (9,1)--(10,2)--(12.5,1)--(10,0)--(9,1);
\draw (10,0)--(10,2)--(11,1.3)--(11,0.7)--(10,0)--(11,1.3);
\draw (11,1.3)--(12.5,1)--(11,0.7);

\draw (10,-1.5) node {$E_6^{(31)}$};

\filldraw[color=black] 
(14,1) circle [radius=2pt]
(15,1) circle [radius=2pt]
(16,2) circle [radius=2pt]
(16,0) circle [radius=2pt]
(17,1)  circle [radius=2pt]
(18,1) circle [radius=2pt];

\draw (14,1)--(16,2)--(18,1)--(16,0);
\draw (14,1)--(15,1)--(16,2)--(17,1)--(18,1);
\draw (15,1)--(16,0)--(17,1);
\draw (14,1)--(16,0)--(16,2);
\draw (15,1)--(17,1);
\draw (16,-1.5) node {$E_6^{(32)}$};

\end{tikzpicture}
}}

\caption{The $32$ forbidden graphs.}
\label{E6graphs}
\end{figure}

This has the following implication (see also \cite{fullclass}):

\begin{theorem}\label{generators}
Let $\mathcal{G}$ be a set of  Pauli strings in $\mathfrak{su}(2^n)$ with $n\geq 3$
with connected frustration graph $\Gamma$.
Let $W$ be the subspace of $V$ spanned by $\{v_p\mid p\in \mathcal{G}\}$,
and $Q_W$ the quadratic form $Q$ restricted to $W$ and $R$ the radical of $Q_W$ in $W$ with dimension $r$. 
Then for $\mathfrak{g}$ the Lie subalgebra generated by $\mathcal{G}$ we  have one of the following:

\begin{enumerate}
\item
$\mathfrak{g}$ is isomorphic to  $\bigoplus_{i=1}^{2^r}\mathfrak{su}(2^k)$ for $k\geq 3$ if and only 
if  $\Gamma$ contains one of the subgraphs of \cref{E6graphs}
and $W/\mathrm{Rad}(Q_W)$ has odd dimension $2k+1$.
\item
$\mathfrak{g}$ is isomorphic to  $\bigoplus_{i=1}^{2^r}\mathfrak{so}(2^k)$  for $k\geq 3$
if $\Gamma$  contains one of the subgraphs of \cref{E6graphs}
and $W/\mathrm{Rad}(Q_w)$ is of $+$-type with dimension $2k$.
\item
$\mathfrak{g}$ is isomorphic to  $\bigoplus_{i=1}^{2^r}\mathfrak{sp}(2^k)$  for $n\geq 3$
if  $\Gamma$  contains one of the subgraphs of \cref{E6graphs} and
$W/\mathrm{Rad}(Q_w)$ is of $-$-type with dimension $2k$.
\item
$\mathfrak{g}$ is isomorphic to  $\bigoplus_{i=1}^{2^r}\mathfrak{so}(k)$  for $k\geq 3$
if and only if $\Gamma$  is the line graph of a multi-graph on $k$ vertices,
with $4\not\mid k$.
\item $\mathfrak{g}$ is isomorphic to  $\bigoplus_{i=1}^{2^{r'}}\mathfrak{so}(k)$  for $k\geq 3$ 
if and only if $\Gamma$   is the line graph of a multi-graph on $k$ vertices with $ 4\mid k$,
where $r'=r-1$ or $r$.
\end{enumerate}
\end{theorem}

\begin{remark}\label{nonstandard}
Notice that in the last case we have two possibilities for $r'$.
This can be seen as follows.
Suppose $\Gamma$ is the line graph of the multi-graph $\Delta$ with vertex set $\Omega$ of size $k>4$ but with   $4\mid k$.

Let $\Delta_0$ be a spanning tree of $\Delta$.
The edges of $\Delta_0$ are vertices of $\Gamma$ generating a subspace
$S_0$ of $S_\mathfrak{g}$ isomorphic to $T(\Omega)$ embedded in a subspace $V_0$ of $V$.

If this embedding is the quotient of the standard embedding, then $\mathfrak{g}$ is isomorphic to  $\bigoplus_{i=1}^{2^{r}}\mathfrak{so}(k)$, as as each dimension of $\mathrm{Rad}(Q_W)$ doubles 
the number if copies of $\mathfrak{so}(k)$.

However, if this embedding is  the standard embedding and $V_0\cap S_\mathfrak{g}=S_0$ , then $\mathfrak{g}$ is isomorphic to  $\bigoplus_{i=1}^{2^{r-1}}\mathfrak{so}(k)$,
as the radical of $Q_{V_0}$ does not double the number of ideals.
But, inside $V_0$, there is still room for  extra points of $S$.
Indeed, if $S_0$ is a proper subset of $V_0\cap S_\mathfrak{g}$, then 
$\mathfrak{g}_{V_0\cap S_\mathfrak{g}}$ is $\mathfrak{so}(k)\oplus \mathfrak{so}(k)$
and  $\mathfrak{g}$ is isomorphic to  $\bigoplus_{i=1}^{2^{r}}\mathfrak{so}(k)$.
\end{remark}

\begin{example}[Dynamical Lie algebras related to Quantum Approximate Optimization Algorithm]
\label{freeLie}
Let $\Gamma$ be a connected graph with vertex set $\{1,\dots,n\}$. If vertex $j$ is adjacent to $k$, then we write $j\sim k$.

Let $\mathfrak{g}$ be generated by
$$iX_\ell,iZ_jZ_{k}$$
where $ \  1\leq \ell \leq n,\text{ and}\ \ 1\leq j,k\leq n-1, j\sim k$.

This is the dynamical Lie subalgebra of $\mathfrak{su}(2^n)$ related to a free Ansatz to the  Quantum Approximate Optimization Algorithm for finding an approximate solution to the maxcut problem for the graph $\Gamma$. See \cite{freeLie}.

The frustration graph on the Pauli strings is 
then the incidence graph of $\Gamma$, i.e., the graph with as vertices the vertices and edges of 
$\Gamma$ and adjacency defined by incidence of  such a vertex with an edge.  

If there exists a vertex $v$ with valency three in $\Gamma$, i.e, $\Gamma$ contains a claw,
then the incidence graph will contain a subdiagram $E_6$. See \cref{fig:incidence}.

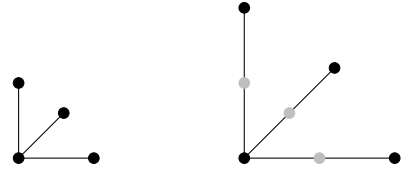
\begin{figure}
\begin{tikzpicture}
\filldraw[color=black]
(0,0) circle [radius=2pt]
(0.6,0.6) circle [radius=2pt]
(1,0) circle [radius=2pt]
(0,1) circle [radius=2pt];

\draw (0,0)--(1,0);
\draw (0,0)--(0,1);
\draw (0,0)--(0.6,0.6);

\filldraw[color=black]
(3,0) circle [radius=2pt]

(5,0) circle [radius=2pt]

(4.2,1.2) circle [radius=2pt]

(3,2) circle [radius=2pt];

\draw (3,0)--(5,0);
\draw (3,0)--(4.2,1.2);
\draw (3,0)--(3,2);

\filldraw[color=lightgray]
(3.6,0.6) circle [radius=2pt]
(3,1) circle [radius=2pt]
(4,0) circle [radius=2pt];

\end{tikzpicture}
\caption{A claw and its incidence graph
with edges in gray.}
\label{fig:incidence}
\end{figure}
If not, then $\Gamma$ is a path or cycle
and $\mathfrak{g}$ is isomorphic to a direct sum of copies of
$\mathfrak{so}(2k)$ for some fixed $k$.

Suppose $\Gamma$ is a path with $n$ vertices, then its incidence graph
is a path with $2n-1$ vertices.
But then $S_{\mathfrak{g}}$ is isomorphic to $T(\Omega)$ for some set $\Omega$
of size $2n$, implying $\mathfrak{g}$ to be isomorphic to $\mathfrak{so}(2n)$.

Now consider the Lie subalgebra $\hat{\mathfrak{g}}$ generated by $\mathfrak{g}$
and the additional generator $iZ_1Z_n$.
So, we consider the dynamical Lie algebra of the cycle graph on $n$ vertices from \cite{freeLie}. As frustration graph on the $2n+1$
generators we obtain the cycle graph on $2n+1$
 vertices, which is (isomorphic to) its own line graph.

The extra generator $iZ_1Z_n$ maps to the vector
$$z_{1,n}=z_1+z_n+r\in V,$$
which is not in $S_\mathfrak{g}$.
(Here we use the basis $$x_1,\dots, x_n,z_1,\dots, z_n,r$$ as in \cref{basis}.)
Indeed, with $$z_{j,j+1}=v_{iZ_jZ_{j+1}}=z_j+z_{j+1}+r$$ and $\delta=0$ if $n$ is odd, and $1$ otherwise,  we find
$$\begin{array}{ll}
z_{1,n} &=z_{1,2}+z_{2,3}+\cdots +z_{n-1,n}+\delta r\\
        &=(x_1+r)+z_{1,2}+\cdots+ (x_{n-1}+r)+\\
        &+z_{n-1,n}+r+x_1+\cdots +x_n\\
        &=v_{iX_1Z_1Z_2\cdots X_{n-1}Z_{n-1}Z_n}+x_1+\cdots +x_n.\\
\end{array}$$

The vector $x_1+x_2+\dots+x_n$ is isotropic and perpendicular to all vectors $x_i$ and $z_{j,j+1}$.
So, $iZ_1Z_n$ and $iX_1Z_1Z_2X_2\cdots Z_{n-1}Z_n$ are in relation $\equiv$.
We find $S_{\hat{\mathfrak{g}}}$ to be isomorphic to $T(\Omega,\Omega')$
where $|\Omega|=2n$ and $|\Omega'|=1$.
In particular, we find the Lie subalgebra $\hat{\mathfrak{g}}$ to be isomorphic to  $\mathfrak{so}(2n)\oplus \mathfrak{so}(2n)$.

Now suppose $\Gamma$ is a connected graph on $n$ vertices 
which is not a path nor a cycle and the generators
are $iX_j$ with $1\leq j\leq n$ and $iZ_jZ_k$ where $\{j,k\}$ is an edge,
then let $\mathfrak{g}$ be the Lie subalgebra  generated by these Pauli strings.
Clearly with respect to the basis $x_j,z_j,r=i\one$
we find that all elements of $S_\mathfrak{g}$ are of the form 
$\sum_{i=1}^n\alpha_i x_i+\sum_{j=1}^n\beta_jz_j+\gamma r$
with $\sum_{i=1}^n\beta_i+\gamma=0$.

If the graph is bipartite with the vertex set split into $I\cup J=\{1,\dots,n\}$
we even find them to also satisfy  $\sum_{i=1}^n\alpha_i+\sum_{i\in I}\beta_i+\gamma=0$.

As for such  a bipartite graph we find $S_\mathfrak{g}$
to contain the vectors
$$\begin{array}{ll}
x_i+r & \text{ with } 1\leq i\leq n\\
z_i+z_j+r &\text{ with } i\in I, j\in J \text{ and } i\sim j\\
\end{array}$$
we find $V_{S_\mathfrak{g}}$ of codimension  $2$.

Indeed, if $i_1, i_2\in I$ are at distance $2$ with common neighbour $j$, 
then $$z_{i_1}+z_{i_2}=(z_{i_1}+z_j+r)+(z_{i_2}+z_j+r)\in V_{S_\mathfrak{g}}$$
and hence by induction and symmetry of the argument 
$$z_{k_1}+z_{k_2}\in V_{S_\mathfrak{g}}$$
for any $k_1,k_2\in J$ and $k_1,k_2\in J$.
So, $$\langle z_i+z_j+r \mid i\in I, j\in J \text{ and } i\sim j\rangle$$
has dimension at least $n-2$.

Similarly, we find that in case $\Gamma$ is not bipartite
$V_{S_\mathfrak{g}}$ is of codimension $1$
and equals the hyperplane
$\sum_{i=1}^n\beta_i+\gamma=0$, which is the hyperplane
$(x_1+\cdots x_n)^\perp$.
The quadratic form $Q$ restricted to this hyperplane has
as radical $\langle x_1+\cdots +x_n\rangle$.
So, for a non-bipartite $\Gamma$ we find
$\mathfrak{g}$ to be isomorphic to $\mathfrak{su}(2^{n-1})\oplus \mathfrak{su}(2^{n-1})$.
If $\Gamma$ is bipartite, we have $V_{S_\mathfrak{g}}$ a codimension $2$
space contained in $(x_1+\cdots +x_n)^\perp$, and find
$\mathfrak{g}$ to be isomorphic to one of
$$\begin{array}{l}
\mathfrak{sp}(2^{n-2})\oplus \mathfrak{sp}(2^{n-2})\\
\mathfrak{s0}(2^{n-1})\oplus \mathfrak{so}(2^{n-1})\\
\mathfrak{su}(2^{n-1})
\end{array}$$
depending on the sizes of $I$ and $J$ in the partition $\{1,\dots,n\}=I\cup J$.
See \cite{freeLie}.

\end{example}

\section{An algorithm to determine the isomorphism type of a Pauli Lie algebra} 
\label{sect:alg}

The above  \cref{generators} can also be used to determine the isomorphism type of the 
Lie subalgebra $\mathfrak{g}$ of $\mathfrak{su}(2^n)$ generated by a set $\mathcal{G}$
of Pauli strings.

Suppose $\mathcal{G}$ is a  set of Pauli strings in $\mathfrak{su}(2^n)$ and let $(V,Q)$ the quadratic space pf \cref{sect:pauli}.
Then  the following steps will provide the isomorphism type of the subgeometry $S$ of $NO(V,Q)$ generated by the elements 
$\{v_p\mid p\in \mathcal{G}\}$ and hence also the isomorphism type of  the  Lie subalgebra  of $\mathfrak{su}(2^n)$  generated by $\mathcal{G}$:
\begin{enumerate}[1.] 
\item {\bf Reduce to connected frustration graphs}. Determine the frustration graph $\Gamma$; if 
$\Gamma$ is not connected, then $\mathfrak{g}$ is the direct sum of the Lie subalgebras
generated by the vertices of the various connected components of $\Gamma$.
So, for the next steps we can restrict our attention to the case that $\Gamma$ is connected.
\item {\bf The frustration graph is a line graph}.
Determine whether $\Gamma$ is a line graph of a multi-graph $\Delta$.
Here we can use the algorithm presented in \cite{cuypers_whitney}.
If so, then the subgeometry of $NO(V,Q)$ spanned by the vertices of $\Gamma$
is isomorphic to $T(\Omega,\Omega')$ where $\Omega$ is the vertex set of
$\Delta$ and $\Omega'$ is a set disjoint of $\Omega$. We can take $\Omega$ to be not of size $4$.

If $|\Omega|$ is not divisible by $4$, 
then $|\Omega'|$ equals the dimension of the radical of $Q_W$, where $W$ is the subspace of $V$ spanned by $S$. This dimension can by found using standard linear algebra algorithms.

So, assume $|\Omega|$ is not divisible by $4$.
We still have to determine the size of $\Omega'$.
For this we fix a subset of $\mathcal{G}_0\subseteq \mathcal{G}$ for which
$\{v_p\mid p\in \mathcal{G}_0\}$ is a minimal set of generators
for a subgeometry $S_0$ isomorphic to $T(\Omega)$.
We can take for $\mathcal{G}_0$ the edges of a spanning tree of $\Delta$.
Let $W_0$ be the subspace of $W$ spanned by $S_0$.
Now the set $\Omega'$ will have size equal to the codimension $d$ of $W_0$
in $W$ if $S\cap V_0=S_0$, or $d-1$ if $S\cap V_0$ contains $S_0$ as a proper subspace.
See \cref{nonstandard}.
The Lie subalgebra  
$\mathfrak{g}$ is then isomorphic to  $\bigoplus_{i=1}^{2^{|\Omega'|}}\mathfrak{so}(|\Omega|)$.

\item {\bf The frustration graph is not a line graph}.
For each connected component $\Gamma$ 
which is not a line graph, 
we find the subgeometry of $NO(V,Q)$ spanned by the vertices of $\Gamma_0$
is isomorphic to $NO(W,Q_W)$ where $W$ is the subspace of $V$ which is spanned by 
$\{v_p\mid p $ a vertex of $\Gamma_0\}$.
If $r$ is the dimension of the radical of $(W,Q_W)$  and $m$ the dimension of $W/\mathrm{Rad}(Q_W)$,
then the Lie algebra generated by the elements of $\Gamma_0$ is the  sum of $2^r$ copies of $\mathfrak{su}(2^{\frac{m-1}{2}})$
if $m$ is odd, of copies of $\mathfrak{so}(2^{\frac{m}{2}})$ if the type of 
the form $Q$ induced on  $W/\mathrm{Rad}(Q_W)$ is $+$ and of copies of
$\mathfrak{sp}(2^{\frac{m}{2}-1})$ if the type induced on $W/\mathrm{Rad}(Q_W)$ is $-$.
\item  {\bf Conclusion}.
The isomorphism type of 
$\mathfrak{g}$ is then the  direct sum of the isomorphism type of the Lie subalgebras determined by the various connected  subgraphs of $\Gamma$.
\end{enumerate}

Each of these four steps can be turned into an algorithm with complexity at most $\mathcal{O}(\max(n,|\mathcal{G}|)^3)$.

For step (1) and (4) this is obvious. 

For step (2)
we can use the results of \cite{cuypers_whitney}. There an algorithm is provided 
that determines whether a connected graph on $|\mathcal{G}|$ vertices is a line graph or not and provides, in case it is a line graph 
a (and then also  actually all) root graph $\Delta$
whose line graph is $\Gamma$. 
This algorithm has complexity $\mathcal{O}(\max(n,|\mathcal{G}|)^3)$, implying   step (2) has complexity at most $\mathcal{O}(|\max(n,|\mathcal{G}|)^3)$ in case the number of vertices of
$\Delta$ 
is not divisible by $4$. The extra steps needed in case the number of  vertices in $\Delta$ is divisible by $4$ can also be bounded by $\mathcal{O}(\max(n,|\mathcal{G}|)^3)$.

It remains step (3).
If $\Gamma$ is not the line graph of a multi graph we have to determine
the type of the subspace $W$ of $V$ spanned by the vertices of $\Gamma$.
We can consider the subspace $W$ spanned by the vectors $v_p$ of $V$, where $p\in \mathcal{G}$.
A straightforward Gram-Schmidt computation on these vectors will provide a \emph{hyperbolic basis}
$$v_1,w_1,\dots,v_k,w_k,r_1,\dots, r_l$$ of $W$
with $v_1,w_1\dots,v_{k},w_k$ being a hyperbolic basis for the space $W_0$ spanned by these vectors 
and the remaining vectors $r_{1},\dots r_\ell$ a basis for the radical of $f_W$, the restriction of $f$ to $W$.
Evaluating $Q$ on the hyperbolic $2$-spaces $\langle v_i,w_i\rangle$ with $1\leq i\leq k$, will tell us the type of $Q$ restricted to the $2$-spaces and hence also to $W_0$.

Finally if $Q(r_i)=0$ for all $i$, then the radical of  $Q_W$ has dimension $l$, and the type of $Q_w$ is determined as
the product of types of the various $2$-spaces $\langle b_{2i-1},b_{2i}\rangle$, with $1\leq i\leq k$.
If there is an $i$ with $1\leq i\leq r$ and  $Q(r_i)=1$, we find 
$\mathrm{Rad}(Q_W)$ to have dimension $r-1$ and the form induced by $Q$ on $W/\mathrm{Rad}(Q_W)$ to be of $0$-type.

In this way we obtain not only the type of $Q$ on the subspace of $W$ but also 
the isomorphism type of the Lie subalgebra $\mathfrak{g}$. 
So, we have found an algorithm that determines the isomorphism type of $\mathfrak{g}=\langle \mathcal{G}\rangle$ with complexity
$\mathcal{O}(\max(n,|\mathcal{G}|)^3)$.

For more details see \cref{app:alg}.

\section{Relations to recent work}

The connection between dynamical Lie algebras generated by Pauli strings and the partial linear spaces  
$NO(V,Q)$ provides us with a rich geometric model for studying 
these dynamical Lie algebras.
In this section we provide some examples of recent research, where the use of this connection 
provides (at least in our opinion) extra insight and options for more general or stronger results,  and simplifies proofs.

\subsection{Subalgebras generated by Pauli strings}
In \cite{fullclass} the authors obtain the same classification for Lie
subalgebras of $\mathfrak{su}(2^n)$ generated by Pauli strings as in \cref{classification}.
Their approach is somewhat different. Starting with a minimal set of generators,
they construct a new set of generators having a frustration graph which is a tree.
For all possible trees they provide Pauli generators in $\mathfrak{su}(2^k)$, $\mathfrak{sp}(2^k)$,
$\mathfrak{so}(2^k)$ and $\mathfrak{so}(k)$  having this tree as 
frustration graph. 

In our setting, the construction of the new set of generators for a Lie subalgebra $\mathfrak{g}$ from a starting set of $\mathcal{G}$ of generating Pauli strings
is equivalent to finding a new generating set for $S_\mathfrak{g}$.
This approach has already been investigated by Brown and Humphries \cite{brown}.
This method lies also at the heart of our characterization of line graphs in \cref{E6theorem}
as given in \cite{cuypersgraphs}, leading to the $32$ forbidden graphs given in \cref{E6graphs}.

Although the approach of \cite{fullclass} can also be turned into an algorithm, we think our approach provides
a more natural way to classify the Lie subalgebras generated by Pauli strings and determine them algorithmically.
Not only does the set up of \cite{fullclass} require that the set of generators is minimal in the sense that any proper subset will not generate the same Lie subalgebra, it also does give less information on how disconnected frustration graphs can give rise to
various classes of Lie subalgebras under the action of the Clifford group.

In two cases we find that a simple Lie subalgebra $\mathfrak{g}$ of $\mathfrak{su}(2^n)$ generated by Pauli strings
gives rise to a subspace $W=V_\mathfrak{g}$ of $V$ such that $(V,Q_W)$ admits a nontrivial
radical for the associated symplectic form $f_W$.
Indeed, for natural subspaces of $NO(V,Q)$ which modulo the radical of $Q$ are
of odd dimension, and for the standard embedding of spaces $T(\Omega,\Omega')$ with $|\Omega|$
being of even size.
These are precisely the cases described in \cite{fullclass} to be related to algebraic dependencies between the generating Pauli strings. 

In \cite{wiersema1,wiersema2} the authors classify Lie subalgebras
of $\mathfrak{su}(2^n)$ with $2$-local action.
They start with the so-called interaction graph whose vertex set is $\{1,\dots,n\}$
and an arbitrary set $E$ of edges.
They fix a subset  $\mathcal{X}$ of Pauli strings from $\mathfrak{su}(4)$
and consider for the each edge $\{i,j\}$ the 
Pauli strings $$I\otimes\cdots \otimes I\otimes P_i\otimes I\otimes \cdots \otimes I\otimes Q_j\otimes I \cdots \otimes I$$
where $P_i,Q_j\in \mathcal{X}$ are at place $i$ and $j$, respectively, in this tensor product.
Several of the results of \cite{wiersema1,wiersema2}  can be obtained by translating the problem on Lie subalgebras generated
by these Pauli strings to the problem of finding the subspaces of $NO(V,Q)$ generated by some
subsets of non-isotropic vectors, and vice versa.

These dynamical Lie algebras have also been studied in \cite{freeLie} in connection  
to quantum approximation optimization algorithm.
In \cref{freeLie} we show that the results of \cite{freeLie} on dynamical Lie subalgebras 
related to the free Ansatz of this quantum algorithm follow easily from our work.

\subsection{Minimal sets of generators}

\cref{generators} generalizes several results of \cite{generators}. 
Indeed, as the dimension of $V$ is $2n+1$, we need at least $2n+1$ generating non-isotropic vectors for $NO(V,Q)$ and then also  $2n+1$ Pauli strings to 
generate $\mathfrak{su}(2^n)$. \cref{generators} provides a necessary and sufficient condition that a set of $2n+1$
Pauli strings does generate  $\mathfrak{su}(2^n)$.

\begin{theorem}
Let $\mathcal{G}$ be a set of Pauli strings in $i\mathcal{P}_n$.
Then 
$\mathcal{G}$ generates $\mathfrak{su}(2^n)$ if and only if
the frustration graph $\Gamma_\mathcal{G}$ is connected, contains a subgraph of \cref{E6graphs}, and 
$\{v_p\mid p\in \mathcal{G}\}$  spans $V$. In particular, $|\mathcal{G}|\geq 2n+1$.
\end{theorem}

Our results also give similar results for the Lie algebras $\mathfrak{sp}(2^n)$, $\mathfrak{so}(2^n)$ and $\mathfrak{so}(n)$.

In \cite{optimal} it is shown that the dimension of $W/\mathrm{Rad}(f)$ is an invariant 
for the minimal generating sets of Pauli strings for a Lie subalgebra $\mathfrak{g}$ of 
$\mathfrak{su}(2^n)$ which itself is generated by Pauli strings.
Suppose $\mathcal{G}$ is a set of generators of minimal size for  $\mathfrak{g}$.
It is claimed in \cite[Result 2, Theorem 19]{optimal} that any set $\mathcal{G}_0$
of Pauli's in $\mathfrak{g}$ with the same size as  $\mathcal{G}$
and frustration graph $\Gamma_0$ also generates   $\mathfrak{g}$
if the dimension of $V_{\Gamma_0}/\mathrm{Rad}{f_{\Gamma_0}}$ has the same dimension
as $W/\mathrm{Rad}(f)$.

However, in this characterization, the authors of \cite{optimal} forget to add the connectivity of the corresponding 
frustration graphs.

The condition on the dimension implies that the vectors $v_p$ with $p\in \mathcal{G}_0$
generate the subspace $W$.
So, the only possibility for $\mathcal{G}_0$ not to generate $\mathfrak{g}$
is that its frustration graph is a line graph, but that of $\mathcal{G}$ not.
However, as follows from \cref{generators}, this is not possible.

In \cite{minimalparitybasis} the authors consider minimal generating sets $\mathcal{G}$ for
$\mathfrak{su}(2^n)$ which contain $\mathcal{G}_0=\{iX_i,iZ_i$ where $1\leq i\leq n\}$.
Here minimal means, no proper subset of $\mathcal{G}$ generates $\mathfrak{su}(2^n)$.

These Pauli strings map to vectors of $V$ generating a hyperplane $V_0$ of $V$ which is the orthogonal sum
$$V_0=V_1\perp V_2\cdots \perp V_n$$
of elliptic $2$-spaces.
Let $r$ be the nonzero vector in the radical of $Q$. So, $V$ is spanned by $V_0$ and $r$. 

Now consider an element $p\in \mathcal{G}\setminus\mathcal{G}_0$. Then $v_p$ can be written as a sum
$$v=\lambda r+\sum_{i\in S} v_i$$
for some subset $S$ of $\{1,\dots, n\}$,
where for $i\in S$ we have $v_i$ is a nonzero vector of $V_i$ and $\lambda\in \mathbb{F}_2$.
So, then $Q(v_p)=\lambda+|S|\pmod{2}$ and $v_p\not \in V_0$ if and only if $\lambda\neq 0$.
Moreover, if $Q(v_p)=1$, then $v_p$ is collinear with two nonzero vectors of $V_i$ for all $i\in S$
but to none in $V_i$ for $i\not \in S$.

So, if $|S|>3$ is even, and $\lambda=1$, we find
in the frustration graph of $\mathcal{G}_0\cup \{v_p\}$ a subgraph $E_6$ and
the subspace of $NO(V,Q)$ spanned by $\{V_i,v_p| i\in S\}$ equal to  
be $NO(W,Q_W)$ where $W=\bigoplus_{i\in S} V_i\oplus \langle r\rangle$. 

In particular, if $\lambda=1$ and $S=\{1,\dots,n\}$ is of even size at least $3$,
we find the subspace of $NO(V,Q)$ spanned by $\{V_i,v| i\in S\}$ to be the full
$NO(V,Q)$. The corresponding Pauli strings then generate $\mathfrak{su}(2^n)$.

However, if $n$ is odd, we also find that any non-isotropic vector outside $V_0$ will not provide a 
connected frustration graph.
So, by adding a single Pauli $p$ to our initial set of generators in $\mathcal{G}_0$  we will still
not generate the full $\mathfrak{su}(2^n)$.

This proves Theorem 2 of \cite{minimalparitybasis}.
Also Theorem 1 of \cite{minimalparitybasis} follows easily, and can even be extended a bit.

If for $S\subseteq \{1,\dots,n\}$
we define $$v_{S}=\lambda r+\sum_{i\in S} v_i$$
where $0\neq v_i\in V_i$ and  $\lambda+|S|\pmod{2}=1$,
then for sets $S_1,\dots, S_k$
with 
\begin{itemize}
\item $S_1\cup \cdots \cup S_k=\{1,\dots,n\}$,
\item $|S_i|$ is even for at least one $i\in \{1,\dots, k\}$,
\item the graph of the set $\{1,\dots, k\}$ with $i$ adjacent to $j$ if and only if $S_i\cap S_j\neq \emptyset$ is connected,
\end{itemize}
we find vectors $v_{S_1},\dots, v_{S_k}$ which 
span together with $V_0$ the full space $V$ (as indeed at least one vector has $\lambda=1$).
Moreover, these vectors are all non-isotropic and  the conditions easily imply that the frustration graph is connected and contains a subgraph $E_6$.
Hence  the corresponding Pauli strings together with $\mathcal{G}_0$ generate the full Lie algebra $\mathfrak{su}(2^n)$.

\subsection{Commutator graph}
In \cite{showcase,West} the authors explore average notions
of scrambling, chaos and complexity over ensembles of systems with dynamical Lie algebras that possess a basis consisting of  Pauli strings.
A main tool in these explorations is the commutator graph.

Let $\mathfrak{g}$ be generated by a set $\mathcal{G}$ of Pauli strings in $i\mathcal{P}_n$
different from $1_{2^n}$. 
Then the \emph{commutator graph} is the graph on all Pauli strings in $i\mathcal{P}_n$
where two Pauli strings $p$ and $q$ are adjacent if and only
there is a $g\in \mathcal{G}$ such that 
$$[p,g]=\pm q$$
or equivalently if there is a $g\in \mathcal{G}$ with
$$v_p+v_g=v_q.$$
So, we can also identify the commutator graph
with the graph on the points of $NO(V,Q)$,
where two vertices $v_p$ and $v_q$, with $p,q$ Pauli strings in $i\mathcal{P_n}$ are adjacent if and only
if there is a $g\in \mathcal{G}$ with 
$$v_p+v_g=v_q.$$
(We can leave out $i1_{2^n}$  as it commutes with all Pauli strings
and hence would represent a component consisting of a single vertex inside the commutator graph.)

See \cite{showcase,West} for the definition and applications of the commutator graph.

Notice that if $p$ is adjacent to $q$, and $q=[p,g]$ for some generator $g\in \mathcal{G}$, 
then 
$$[[p,g],g]=[q,g]=p$$ and
$$(v_q+v_g)=(v_p+v_g)+v_g=v_p.$$
We find that the commutator graph is a subgraph of the collinearity graph of $NO(V,Q)$.

If the frustration graph $\Gamma$ on $\mathcal{G}$ connected, then we
find that the connected component  $C$ of the 
commutator graph containing an element of $\mathfrak{G}$  equals (after identification) the subspace $S_\mathfrak{g}$.
Indeed, if $p,q\in \mathcal{G}$ are non-commuting,
then $[p,[p,q]]=q$. So, the vertices of $\Gamma$ are contained in $C$.
Moreover,
if $p$ and $q$  are Pauli strings in $\mathfrak{g}$  
and $g$ a generator in $\mathfrak{G}$ with $[p,g]=\pm q$,
then $v_q=v_p+v_q\in S$.
So, all the vertices of $C$ are contained in $S$.
If $g_1,g_2$ are two non-commuting Pauli strings in $\mathfrak{g}$,
then, if $[p,[g_1,g_2]]\neq 0$ we find (up to a permutation of $g_1$ and $g_2$) that
$[p,g_1]\neq 0$ and $[p,g_2]=0$.
But then $[p,[g_1,g_2]]=[[p,g_1],g_2]$.
So, if $[p,g_1]\in C$ and $g_2\in \mathfrak{G}$, we find 
$[p,[g_1,g_2]]\in C$.
But this implies that  $C$ equals $\{[p,g]\mid g\in \mathfrak{g}\cap \Pi\}$,
In particular, $C=S$.

If $p$ is a Pauli string not in $\mathfrak{g}$,
commuting with all elements of $\mathfrak{G}$, then its components is just $\{p\}$.

Finally, assume  $p$ is a Pauli string not in $\mathfrak{g}$,
which does not commute with some element $g$ of $\mathfrak{G}$.
Let $C$ be  the component of the commutator graph containing $p$.
Then, by the same arguments as above, we find $C$ equals $\{[p,g]\mid g\in \mathfrak{g}\cap P\}$.
But that implies that $C\cup S$ is contained in $S_1$ of $NO(V,Q)$ generated by $p$ and $\mathcal{G}$.

Suppose $S=T(\Omega)$ in its standard embedding in $NO(V,Q)$.
Then we can identify $V$ with the subspace of $\mathbb{F}_2\Omega$ of even weight vectors.
We easily check that if $p\in C$ is a vector of weight $2m$, with $m\geq 1$ odd, then the vectors in $C$ are precisely all vectors of weight $2m$.

If $4\mid |\Omega|$ and the embedding of $\Omega$ is  the quotient of the standard embedding, we find the components
to be the images of the components of the standard embedding in this quotient.

If $S=NO(W,Q_W)$ for some subspace $W$ of $V$ with nondegenerate form $Q_W$, where $W$ has dimension at least $7$ or $6$ and is of $-$-type, then 
of course we find besides $S$ also the non-isotropic vector in the radical of $f_W$, if it exists, as a component.
And for a Pauli string with $v_p\not S$ and not perpendicular to $S$, we find $C$
to be contained in the subspace $\hat{W}=\langle W,v_p\rangle$.
Now every elliptic line in $\hat{W}$ will meet $W$ in a point of $S$.
But then, as the complement of $S$ in the collinearity graph of $NO(\hat{W},Q_{\hat{W}})$ is connected (just check inside all dual affine planes) 
one easily checks that all points of $NO(\hat{W},Q_{\hat{W}})$ are in $C\cup S$, so that $C$ is the complement of $S$ inside 
$NO(\hat{W},Q_{\hat{W}})$.

\begin{remark}
In \cite{West} it is stated that  
the Pauli strings spanning $\mathfrak{g}$ itself will
appear as a single component if $\mathfrak{g}$ is a simple Lie algebra.
This indeed follows from the above.
But it is also claimed that for the various simple summants of $\mathfrak{g}$, their set of Pauli strings will
appear as separate components. See \cite[end of page 3]{West}.
However, as the above shows, this is only the case when the various summants are 
themselves generated by Pauli strings.
Examples where this is not the case occur when the geometry $S_\mathfrak{g}$ is isomorphic
to $N(W,Q_W)$ where $Q_W$ has an isotropic radical.
\end{remark}

\subsection{Free-fermion Hamiltonians}

In \cite{Chapman2020} the relation between the frustration graph and free-fermion descriptions of a Hamiltonion are investigated. As these systems on $n$ particles correspond to the case that the dynamical Lie algebra is an orthogonal Lie algebra $\mathfrak{so}(n)$,  \cref{generators}
provides detailed solutions to  problems considered in \cite{Chapman2020}.

\subsection{Cartan decomposition}
A Cartan decomposition of a Lie algebra $\mathfrak{g}$ is a decomposition $\mathfrak{g}=\mathfrak{l}\oplus \mathfrak{m}$
such that

$$[ \mathfrak{l},\mathfrak{l} ]  \subseteq \mathfrak{l}$$
$$[ \mathfrak{m},\mathfrak{m} ]  \subseteq \mathfrak{l}$$
$$[\mathfrak{l},\mathfrak{m} ] \subseteq \mathfrak{m} $$

Cartan decompositions can be very useful in the study of dynamical Lie algebras, see for example \cite{cartan}.
The Lie subalgebras under consideration do admit such Cartan decompositions, which we can easily recognize in the geometry of $(V,Q)$.

Indeed, if $\mathfrak{g}$ is a subalgebra of $\mathfrak{su}(2^n)$ generated by Pauli strings,
then fix  a hyperplane $H$ of $V$ meeting $S_\mathfrak{g}$ non-trivially.
Then for $p,q\in S_\mathfrak{g}$ we find the following:
If $v_p,v_q\in H$ then $[p,q]=0$ or $v_{[p,q]}=v_p+v_q\in H$.
If $v_p,v_q\not\in H$ then $[p,q]=0$ or $v_{[p,q]}=v_p+v_q\in H$.
And finally, $v_p\in H$ but $v_q\not\in H$ then $v_{[p,q]}=v_p+v_q\not\in H$.
So if $\mathfrak{l}$ is the linear subspace of $\mathfrak{g}$ spanned by the Pauli strings $p$ with $v_p\in H$
and $\mathfrak{m}$ the subspace spanned by the Pauli strings $p$ with $v_p\not \in H$,
then $\mathfrak{g}=\mathfrak{l}\oplus \mathfrak{m}$
is a Cartan decomposition.

\section{Discussion}

In this paper we provided a one-to-one correspondence between Lie subalgebras of $\mathfrak{su}(2^n)$
generated by Pauli strings and subspaces of the point-line geometry of non-isotropic points and elliptic lines for a non-degenerate quadratic $\mathbb{F}_2$ space $(V,Q)$ of dimension $2n+1$.
Using this correspondence we are able to not only classify the Lie subalgebras generated by Pauli strings, but also
provide a way to recognize them algorithmically up to isomorphism.

Indeed, if the Lie subalgebra $\mathfrak{g}$ is generated by set $\mathcal{G}$ of Pauli strings with a connected frustration graph,
then $\mathfrak{g}$ is a direct sum of Lie algebras all isomorphic to $\mathfrak{su}(2^m)$, $\mathfrak{so}(2^m)$  or  $\mathfrak{sp}(2^m)$
provided the frustration graph contains one of $32$ subgraphs on $6$ points, or equivalently, a subset of $6$ elements in $\mathcal{G}$ generating 
$\mathfrak{sp}(2^2)$.
If no such subgraphs (or subsets of Pauli strings) are present, then  $\mathfrak{g}$ is a direct sum of Lie algebras isomorphic to
$\mathfrak{so}(m)$ for some fixed $m$.

We also described an algorithm that given $\mathcal{G}$ determines the isomorphism type of $\mathfrak{g}$.
This algorithm runs in time $\mathcal{O}(\max(n,|\mathcal{G}|)^3)$ and for the greater part
depends on some standard linear algebra algorithms for vector spaces over $\mathbb{F}_2$.

The results of our work imply several of the recent results found in \cite{fullclass, wiersema1, wiersema2, freeLie,optimal, generators}.
The geometry $NO(V,Q)$ provides a richer and more natural  geometry underlying the 
Pauli group than the commonly used model of a symplectic space over $\mathbb{F}_2$.
As such it may have impact on the various applications mentioned in these papers
and future investigations of dynamical Lie algebras generated by Pauli strings, or even
by special sums of Pauli strings, see for example \cite{freeLie,reductions,showcase} or \cite{West}.

\newpage


\bibliographystyle{plain}

\bibliography{physics.bib}

{\small

\parindent=0pt
Hans Cuypers\\
Department of Computer Science and Mathematics\\
Eindhoven University of Technology\\
P.O. Box 513 5600 MB, Eindhoven\\
The Netherlands\\
}

\newpage


\appendix

\section{Quadratic spaces over the field with $2$ elements}
\label{app:quad}

In this appendix we provide a short introduction to quadratic spaces over the field $\mathbb{F}_2$.

Let $V$ be a vector space over $\mathbb{F}_2$, the field with $2$ elements.
Then a quadratic form on $V$ is
a map
$$Q:V\rightarrow \mathbb{F}_2$$
such that the form $$f:V\times V\rightarrow \mathbb{F}_2$$
defined by
$$f(v,w)=Q(v+w)+Q(v)+Q(w)$$
for all $v,w\in V$ is a \emph{symplectic form}, i.e. it is bilinear and $f(v,v)=0$ for all $v\in V$

We call $f$ the symplectic form associated to $Q$.

The \emph{radical} of a symplectic form $f$ is the subspace $\mathrm{Rad}(f)=\{v\in V\mid f(v,w)=0$ for all $w\in V\}$.
The form $f$ is called non-degenerate if and only if the radical equals $\{0\}$.
The \emph{radical} of $Q$ denoted by $\mathrm{Rad}(Q)$ is  the subspace  $\{v\in \mathrm{Rad}(f)\mid Q(v)=0\}$ of $\mathrm{Rad}(f)$. The radical  $\mathrm{Rad}(Q)$ of $Q$ is also called the \emph{isotropic} radical.
The form $Q$  is called \emph{non-degenerate} if $\mathrm{Rad}(Q)=\{0\}$ and \emph{degenerate} otherwise.
Then $\mathrm{Rad}(Q)$ has codimension at most $1$ in $\mathrm{Rad}(f)$.

We call a pair $(V,Q)$ with $V$ a vector space equipped with a quadratic form $Q$ a \emph{quadratic space} or \emph{orthogonal space}.

\begin{example}\label{1-spaces}\label{2-spaces}
Let $(V,Q)$ be a quadratic space over the field $\mathbb{F}_2$.
A vector $v\in V$ is called \emph{isotropic} if $Q(v)=1$, and non-isotropic otherwise.

If $W$ is a subspace of $V$ of dimension $2$, then  one of the following can occur:
\begin{itemize}
\item $Q$ restrict to $W$ is $0$.
We call  $W$ an \emph{isotropic} line (or 2-space).
\item $f$ is $0$ on $W$, but $Q$ not.
Then two of the three nonzero vectors of $W$ are mapped to $1$ and one to $0$.
The space $W$ a called a \emph{tangent} line (or 2-space). 
\item $f$ is non-trivial on $W$, but there is a $0\neq w\in W$ with $Q(w)=0$.
Then there are two vectors of $W$ on which $Q$ takes the value $0$ and one where it takes the value $1$. We call $W$ a  \emph{hyperbolic} line (or $2$-space) or of \emph{$+$-type}. 
\item $f$ is non trivial on $W$ and  $Q$ takes only the value $1$ on the nonzero vectors in $W$.
We call  $W$ an \emph{elliptic} line (or $2$-space) or of $-$-\emph{type}. 
\end{itemize}
\end{example}

If $(V,Q)$ is a quadratic space and $V=V_0\oplus V_1$ with $f(v_0,v_1)=0$ for all $v_0\in V_0$ and $v_1\in V_1$, then we say 
$V$ is an \emph{orthogonal sum} of $V_0$ and $V_1$ and we write $V=V_0\perp V_1$.

Nondegenerate quadratic $\mathbb{F}_2$-spaces can be classified:

\begin{theorem}
Let $(V,Q)$ be a finite dimensional non-degenerate quadratic space over $\mathbb{F}_2$.
Then we have one of the following:
\begin{itemize}
\item $V$ has even dimension $2m$ and is the orthogonal sum of  hyperbolic $2$-spaces and an even number of elliptic $2$-spaces. 
\item $V$ has even dimension $2m$ and is the orthogonal sum of $m-1$ hyperbolic $2$-spaces and an odd number of elliptic $2$-space.
\item $V$ has odd dimension $2m+1$ and is the orthogonal sum of $m$ hyperbolic $2$-spaces and a non-isotropic $1$-space which is the radical of $f$.
\end{itemize}

A subspace of $V$ which is the orthogonal sum of an even number of elliptic $2$-spaces
is also the orthogonal sum of an even number of hyperbolic $2$-spaces.
\end{theorem}

Consider a quadratic $\mathbb{F}_2$-space $(V,Q)$ and 
its quotient space $(\overline{V},\overline{Q})$, where $\overline{V}=V/\mathrm{Rad}(Q)$
and $\overline{Q}$ the quadratic form induced by $Q$ on $\overline{V}$.
We call the  quadratic $\mathbb{F}_2$-space $(V,Q)$ \emph{hyperbolic}  or of \emph{$+$-type}
if $(\overline{V},\overline{Q})$ is the orthogonal sum of hyperbolic $2$-spaces.
If $(\overline{V},\overline{Q})$ is the  orthogonal sum of hyperbolic $2$-spaces and one elliptic one, then $(V,Q)$ is called an \emph{elliptic} or of \emph{$-$-type}.

\section{The partial linear space $NO(V,Q)$ and its subspaces}
\label{app:cotriangular}

A \emph{partial linear space} is a pair $(P,L)$, where $P$ is a set of \emph{points} and $L$ a set of subsets of $P$ of size at least $2$ called \emph{lines}, such that
any two points are on at most one line.

Given an quadratic $\mathbb{F}_2$-space $(V,Q)$  we define 
$NO(V,Q)$ to be the partial linear space point set consists of all $v\in V\setminus\mathrm{Rad}(f)$ with $Q(v)=1$
and whose lines are the triples $\{u,v,w\}$ of nonzero vectors in elliptic lines of $V$.
Here $f$ is the symplectic form associated to $Q$.

These geometries are well studied, 
see for example 
\cite{brown2,hall,Hall3,seven,ward}.

In particular, the results of \cite{hall} imply:
\begin{theorem}
Let $S$ be a connected subspace of $NO(V,Q)$, then 
we have one of the following.

\begin{enumerate}
\item There is a subspace $W$ of $V$ such that $S$ equals $NO(W,Q_W)$;
\item There are disjoint sets $\Omega$ and $\Omega'$  such that 
$S$ is isomorphic to a standard  embedding of $T(\Omega,\Omega')$ or its quotient in $V_S$.
\end{enumerate}
\end{theorem}

\begin{proof}
Let $W$ be the subspace of $V$ spanned by the elements of $S$.
 By \cite[Theorem 5.5]{hall} we have one of the following:
 \begin{enumerate}
\item  There is a symplectic form $f_1$ on $W$ such that $S$ consists of all vectors
outside the radical of $f_1$;
\item
There is a quadratic form $Q_1$ on $W$ such that $S$ consists of all vectors on which
$Q_1$ takes the value 1.
\item  $S$ is isomorphic to $T(\Omega,\Omega')$  for some disjoint sets $\Omega,\Omega'$.
\end{enumerate}
If we are in case (i), then all 2-spaces  of $W$ which do not meet $\mathrm{Rad}(f_1)$
are elliptic. Moreover, $\mathrm{Rad}(f_1)$ is contained in the radical of $f_W$, the restriction of $f$ to $W$.
But then $W/\mathrm{Rad}(f_W)$ is an elliptic $2$-space and $S$ equals $NO(W,Q_W)$.

In case (ii) we find for any vector $w\in W$ that $Q_1(w)=1$ implies $Q(w)=1$.
Furthermore for any two collinear vectors $v,w\in W$ we find $f(v,w)=Q_1(v+w)+Q_1(v)+Q_1(w)$.
As the $S$ linearly spans $W$, we find that  $Q_1$ and $Q$ need to be the same.
\end{proof}

\section{Subspaces and Lie subalgebras}

Let $S$ be a subspace of $NO(V,Q)$
then clearly $\mathfrak{g}_S$ is a Lie subalgebra of $\mathfrak{su}(2^n)$
with $S\subseteq S_{\mathfrak{g}_S}$.

Let $\mathcal{P}_S$ be the set of Pauli strings $p\in i\mathcal{P}_n$ with $v_p\in S$.
Then let $\mathfrak{g}_0$ be the linear subspace of $\mathfrak{su}(2^n)$
spanned by $\mathcal{P}_S$.
Now for any two Pauli strings $p,q\in \mathcal{P}_S$ we find $[p,q]=0$ or $r=[p,q]$ is a Pauli string with
$v_r\in S$ and hence $r\in \mathfrak{g}_0$.
So $\mathfrak{g}_S=\mathfrak{g}_0$.
If $p\in \mathfrak{g}_S$ is a Pauli string with $v_p\not\in S$.
Then $p$ has to be a linear combination of elements of $\mathcal{P}_S$.
But as Pauli strings in $i\mathcal{P}_n$ form a basis for  $\mathfrak{su}(2^n)$,
we find $p$ to be a scalar multiple of some element of $\mathcal{P}_S$, contradicting 
 $v_p\not\in S$.
 
This proves  the second statement of the following proposition. The first statement is obvious.

\begin{proposition}
\phantom{a}
\begin{enumerate}
\item 
If $\mathfrak{g}$ is a Lie subalgebra of $\mathfrak{su}(2^n)$  generated by
Pauli strings, then  $S_\mathfrak{g}$  is a subspace of $NO(V,Q)$.
\item 
If $S$ is a subspace of $NO(V,Q)$, then
$\mathfrak{g}_S$ is a Lie subalgebra of $\mathfrak{su}(2^n)$
 with $S_{\mathfrak{g}_S}=S.$
\end{enumerate}
\end{proposition}

\section{Ideals in the Lie algebras}
\label{app:ideal}

Suppose the  Lie subalgebra $\mathfrak{g}$ of $\mathfrak{su}(2^n)$ is generated by its
set of Pauli strings $\mathcal{P}_\mathfrak{g}\subseteq i\mathcal{P}_n$, and let $S_\mathfrak{g}$ be the corresponding subspace of $NO(V,Q)$.

\begin{proposition}\label{app:ideals}
Suppose $S_\mathfrak{g}$ is connected.
\begin{enumerate}
\item If $\equiv$ is trivial on $\mathcal{P}_\mathfrak{g}$, then $\mathfrak{g}$ is simple.
\item Suppose $p_0\neq q_0\in \mathcal{P}_\mathfrak{g}$.
Then $$p_0\equiv q_0 \Leftrightarrow 0\neq v_{p_0}+v_{q_0}\in \mathrm{Rad}(Q_{V_\mathfrak{g}}).$$
\item 
If $p_0\neq q_0\in \mathcal{P}_\mathfrak{g}$ and $p_0\equiv q_0$ then with $z=p_0q_0$ we find
$$I_z^+=\langle (1+z)p\mid p\in \mathcal{P}\rangle$$ and 
$$I_{z}^-=\langle (1-z)p\mid p\in \mathcal{P}\rangle$$
to be ideals of $\mathfrak{g}$.

Moreover, $$\mathfrak{g}/I_z^{+}\simeq I_z^{+}\simeq I_z^{-}\simeq\mathfrak{g}/I_z^{-}$$
is isomorphic to the $\mathfrak{g}_{S_0}$, where $S_0$ is the intersection of $P$ with $W_0$, a hyperplane in the subspace of $V$ generated by $S_\mathfrak{g}$ not containing $v_{p_0}+v_{q_0}$.
\end{enumerate}
\end{proposition}

\begin{proof}
Suppose $\equiv$ is trivial and $\mathcal{I}$ is a nontrivial ideal.
Suppose $x\in \mathcal{I}$. Then $x$ can be expressed as 
$$x=\sum_{y\in \mathcal{P}_\mathfrak{g}} \lambda_y y$$
with $\lambda_y$ scalar.
The set of all $y\in \mathcal{P}_\mathfrak{g}$ with $\lambda_y\neq 0$ is called 
the support of $x$ and denoted by $\mathrm{supp}(x)$.
Take $x\in \mathcal{I}$ with a support of minimal size.
If the support of $x$ is of size $1$, then by connectedness of $S_\mathfrak{g}$
we find $\mathcal{P}_\mathfrak{g}\subseteq \mathcal{I}$ and $\mathcal{I}=\mathfrak{g}$.

So, assume the support is of size at least $2$.
As $\equiv$ is trivial, we find $y,z$ in the  support of $x$
such that there is a $u$ with $[y,u]=0$ but $[z,u]\neq 0$.
But then $$[x,u]=[\sum_{v\in \mathrm{supp}(x)} \lambda_v v,u]=\sum_{v\in \mathrm{supp}(x),v\neq y} \lambda_v [v,u]$$
is a nontrivial element in $\mathcal{I}$ with smaller support. This contradicts the choice of $x$.
This proves (i).

Suppose $p\equiv q$. Then $[p,q]=0$ and $v_p+v_q$ is isotropic. For each Pauli string  $r$ we have 
$f(v_p,v_r)=f(v_q,v_r)$ and hence $v_p+v_r$ is in the radical of $f$.
Hence $v_p+v_r$ is in the radical of $Q$.

On the other hand, if $v_p+v_r$ is in this radical, we find $f(v_p,v_r)=f(v_q,v_r)$
for all $r$ and we easily deduce $p\equiv q$. This proves (ii).

Finally consider (iii).
Suppose $p\equiv q$ and set $z=pq$. The $z$ commutes with all elements from $\mathfrak{g}$.
So for each $r, s\in S_\mathfrak{g}$ we find
$[r,(1+z)s]=[r,s+zs]=[r,s]+[r,zs]=[r,s](1+z)=(1+z)[r,s]$.
Hence $\mathcal{I}^+$ is an ideal. Similarly we can find $\mathcal{I}^-$ to be an ideal.
\end{proof}

\section{Simple Lie algebras and related geometries}
\label{app:lie}
The above proposition shows that  
$\mathfrak{g}$ is isomorphic to a direct sum of $2^r$
copies of Lie subalgebras, where $r$ is the dimension of the isotropic radical of $W$ in case $S_\mathfrak{g}$ is isomorphic to
$NO(W,Q_W)$ or the size of the set $\Omega'$ in case $S_\mathfrak{g}$ is isomorphic to $T(\Omega,\Omega')$ for disjoint sets $\Omega$ and $\Omega'$. 

So, to identify the Lie subalgebras of $\mathfrak{su}(2^n)$ generate by Pauli strings
we need only consider for the geometries $NO(W,Q_W)$ with $Q$ non-degenerate
and $T(\Omega,\Omega')$, where $|\Omega|\geq 5$ and $\Omega$ is the empty set.

\begin{theorem}\label{identify}
Let $\mathfrak{g}$ be a Lie subalgebra of $\mathfrak{su}(2^n)$ generated by Pauli strings.
\begin{enumerate}
\item If $S_\mathfrak{g}$ is isomorphic to $NO(W,Q_W)$, for some non-degenerate quadratic space  $(W,Q_W)$
of odd dimension $2m+1$, then $\mathfrak{g}$ is isomorphic to $\mathfrak{su}(2^m)$;
\item If $S_\mathfrak{g}$ is isomorphic to $NO(W,Q_W)$, for some non-degenerate quadratic space $(W,Q_W)$
of even dimension $2m$ of $+$-type, then $\mathfrak{g}$ is isomorphic to $\mathfrak{so}(2^m)$;
\item If $S_\mathfrak{g}$ is isomorphic to $NO(W,Q_W)$, for some non-degenerate quadratic space $(W,Q_W)$
of even dimension $2m$ of $-$-type, then $\mathfrak{g}$ is isomorphic to $\mathfrak{sp}(2^m)$;
\item If $S_\mathfrak{g}$ is isomorphic to $T(\Omega,\emptyset)$ for some set $\Omega$ of size at least $5$,
then $\mathfrak{g}$ is isomorphic to $\mathfrak{so}(|\Omega|)$.
\end{enumerate}

\end{theorem}

Let $y\in \Pi_n$. Then let  $\tau_{y}:\Pi_n\rightarrow \Pi_n$
be the map defined by

$$\tau_{y}(p)=-y^2\cdot yp^\top y=\begin{cases} yp^\top y &\mathrm{if}\ Q(v_y)=1\\
-yp^\top y &\mathrm{if}\ Q(v_y)=0\\
\end{cases}$$

Then for $p\in \Pi_n$ we have
$\tau_y(p)=\pm p$.

Moreover, for all non-commuting $p,q$ we have
$$\begin{array}{ll}
\tau_{y}([p,q])&=\tau_{y}(pq)\\
&=-y^2\cdot y(pq)^\top y\\
&=-y^2\cdot yq^\top p^\top y\\
&=y^2\cdot yp^\top q^\top y\\
&=yp^\top yy q^\top y\\
&=[\tau_y(p),\tau_y(q)]\\
\end{array}$$
  
So, $\tau_y$ is an automorphism of the Lie algebra $\mathfrak{g}=\mathfrak{su}(2^n)$ and its
fix points form a Lie subalgebra $\mathfrak{g}_{\tau_y}$.
Moreover, if $\tau_y(p)=-p$ and $\tau_y(q)=-q$, for non-commuting Pauli strings $p,q$,
then $\tau_y(pq)=\tau_y([p,q])=[\tau_y(p),\tau_y(q)]=[-p,-q]=pq$.
So, $S_{\mathfrak{g}_{\tau_y}}$ is a subspace of $NO(V,Q)$ meeting every line in at least a point.
Such subspaces are called geometric hyperplanes of  $NO(V,Q)$.
Every such geometric hyperplane of $NO(V,Q)$ is the intersection
of $NO(V,Q)$ with a hyperplane of $V$.

If $y=\epsilon I$, with $\epsilon \in \{\pm 1,\pm i\}$, 
then $\tau_y(p)=-p^\top$ for all Pauli strings $p$.
So $$\mathfrak{g}_{\tau_y}=\mathfrak{so}(2^n)=\{x\in \mathfrak{su}(2^n)\mid x=-x^\top\}.$$

Now assume $y$ is not a multiple of $I$.
As the Clifford group  is transitive
on the element $y$ with $y^2=1$ and on those with $y^2=-1$, we find that these subalgebras are all isomorphic to
the Lie subalgebra
$\{x\in \mathfrak{su}(2^n)\mid x=-yx^\top y\}$
where $y=Z_1$
or 
the Lie subalgebra
$$\mathfrak{sp}(2^{n-1})=\{x\in \mathfrak{su}(2^n)\mid x=yx^\top y\}$$
where $$y=\begin{pmatrix} 0 & -1_{2^{n-1}}\\ 1_{2^{n-1}} &0\end{pmatrix}.$$
  
In the first case, the equation   $$x+yx^\top y=0$$
where $y=Z_1$ can be transformed into  
$$x+x^\top=0,$$ by applying
the element $S$, where $S=\begin{pmatrix}1&0\\0&i\\\end{pmatrix}\otimes I\otimes \dots\otimes I$ from the Clifford group. Then $Z\overline{S}=S$ and $S^\top Z= \overline{S}^\top$.
For each element $x$ satisfying  $$x+yx^\top y=0$$
we find $x'=\overline{S}^\top xS$ to satisfy 
$$\begin{array}{ll}
0&=Sx'\overline{S}^\top+ y(Sx'\overline{S}^\top)^\top y\\
&=Sx'\overline{S}^\top+y\overline{S}x'^\top S^\top y\\
&=Sx'\overline{S}^\top+Sx'^\top \overline{S}^\top\\
&=S(x'+x'^\top)\overline{S}^\top\\
\end{array}$$
which is equivalent to
$$x'+x'^\top=0.$$

So, we encounter two non-isomorphic Lie subalgebras
$$\begin{array}{ll}
\mathfrak{so}(2^{n})&=\{x\in \mathfrak{su}(2^n)\mid x=-x^\top \}\\
\mathfrak{sp}(2^{n-1})&=\{x\in \mathfrak{su}(2^n)\mid x=yx^\top y\}
\end{array}$$
where $$y=\begin{pmatrix} 0 & -1_{2^{n-1}}\\ 1_{2^{n-1}} &0\end{pmatrix}.$$
 
As the dimension of $\mathfrak{so}(2^{n})$ equals the number of points in 
a geometric hyperplane of $NO(V,Q)$ of $+$-type
and the dimension of $\mathfrak{sp}(2^{n-1})$ the number of points in 
a geometric hyperplane of $NO(V,Q)$ of $-$-type, we have proven the first three statements of \cref{identify}.

\section{Enumerating $T(\Omega)$}
\label{app:triangle}

Suppose $\Delta$ is a connected multi-graph, and $\Delta_0$ a spanning tree
of $\Delta$, which is a connected subgraph on all the vertices of $\Delta$ and
some edges. Such a spanning tree does always exist.
Denote by $\Gamma$ the line graph of $\Delta$ and by $\Gamma_0$ that of $\Delta_0$,

Suppose $\mathcal{G}$ is a set of Pauli strings with frustration graph
$\Gamma$. Let $\mathfrak{g}$ be the subalgebra of $\mathfrak{su}(2^n)$ generated by $\mathcal{G}$.

Then consider the subset $\mathcal{G}_0$ of vertices of $\Gamma_0$.
We can embed $\Gamma_0$ as points in 
$T(\Omega)$, and find it to generate a subspace $S_0$ of $NO(V,Q)$ isomorphic to $T(\Omega)$.

Indeed, if $w$ is a leaf of $\Delta_0$, which is on the edge $e$, then by induction
we find that the points of $\Gamma_0$ different from $e$ generate
$T(\Omega\setminus\{w\})$.
But then one easily checks that all points of $\Gamma_0$ generate the full
space $T(\Omega)$.

In order to enumerate all elements of $S_0$ we can proceed as follows.
Let $\Delta_{i+1}$ be recursively defined as the tree obtained from $\Delta_i$
by removing a leaf end the edge $e_{i}$ on this leaf from $\Delta_i$.

Then denote by $\Gamma_i$ the line graph of $\Delta_i$ and $S_i$ the subspace
of $NO(V,Q)$ generated by the points $v_p$ with $p$ a vertex of $\Gamma_i$.
Then $$S_{n-1}\subsetneq S_{n-2}\subsetneq\dots\subsetneq S_{0}.$$
We find $S_i$ by adding $v_{e_{i}}+v$ to $S_{i+1}$ for all $v\in S_{i+1}$ with $f(v,v_{e_{i}})=1$.
This requires $O(n^3)$ operations.

In this way we can also identify a basis for the Lie subalgebra $\mathfrak{so}(n)$
generated by the elements in $\mathcal{G}_0$ inside $\mathfrak{g}$.

\section{The algorithm}
\label{app:alg}

In this section we will make the algorithm of \cref{sect:alg} explicit.

As input we have a set $\mathcal{G}=\{p_1,\dots, p_m\}$ of Pauli strings in $\mathfrak{su}(2^n)$.
These Pauli strings are given as (minimal) products of a fixed set
of $2n+1$ Pauli strings generating $\mathfrak{su}(2^n)$.
Let $B$ be the basis of $V$ obtained from these generating Pauli strings.

\begin{enumerate}[1.]
\item
The first preliminary step is to include that information in to determine
the vectors $v_{p_j}$ in the basis $B$.
\item
The second step is to determine the Gram matrix $G$ of the vectors $\{v_p\mid p\in \mathcal{G}\}$  
with respect to the form $f$.
This requires $O(m^2\cdot n)$ basic computations.

\item 
Now find the connected components of the graph with adjacency matrix $G$.
A breadth-first search will do this in runtime $O(m^2)$.

\end{enumerate}

We now have found the connected components of the frustration graph of 
the set $\mathcal{G}$.
The next steps will be done for each connected component of the frustration graph.
So, from now on, assume the frustration $\Gamma$ on the vertices of $\mathcal{G}$ is connected.
\begin{enumerate}[1.,resume]

\item
Determine whether $\Gamma$ is a line graph of a multi-graph $\Delta$
with $k\neq 4$ vertices and determine $\Delta$ (Notice this is always possible).
In \cite{cuypers_whitney} one can find an algorithm running in time $O(|E|)$,
where $E$ is the edge set. In particular, it requires  $O(m^2)$ basic calculations.
\item If yes,  determine $\mathtt{dim}$, the dimension of $\langle v_p\mid p\in \mathcal{G}\rangle$ using the Gaussian algorithm.
This can be done in $O(\max(n,m)^3)$ basic calculations.

Take a  a spanning tree $\Delta_0$ of $\Delta$ with edge set $\mathfrak{G}_0$.
Then determine $\dim(W_0)$ where $W_0$ is the subspace of $V$ spanned by $\{v_p\mid p\in G_0\}$.
If $4\not \mid |\Omega|$, then the size of $|\Omega'|$ equals $\dim(W)-\dim(W_0)$.
The Lie subalgebra  
$\mathfrak{g}$ is then isomorphic to  $\bigoplus_{i=1}^{2^{|\Omega'|}}\mathfrak{so}(|\Omega|)$.

If $4\mid \Omega$ and $\dim(W_0)=|\Delta|-2$, we find that $S\cap V_0$
is a quotient of the standard embedding of $T(\Omega)$.
Again, we can conclude that $\Omega'$ has size $\dim(W)-\dim(W_0)$.
and  $\mathfrak{g}$ is then isomorphic to  $\bigoplus_{i=1}^{2^{|\Omega'|}}\mathfrak{so}(|\Omega|)$.

If $4\mid \Omega$ and $\dim(W_0)=|\Delta|-1$, we find that $S\cap V_0$
is  the standard embedding of $T(\Omega)$.
Determine the isotropic vector $r_0$ in the radical of $Q$ restricted to $W_0$.
For each of the generators $p$ in $\mathcal{G}\setminus \mathcal{G}_0$ we can 
check whether $v_p=r_0+v_{p'}$ for some $v_{p'}$ in the subspace $S_0$ generated
by $\{v_q\mid q\in \mathcal{G}_0\}$.
This can be done by enumerating all elements of $S_0$ as in \cref{app:triangle}
and then check for all elements $p$ in $\mathcal{G}\setminus \mathcal{G}_0$ whether $v_p$ appears in $S_0$.
If for some $p$ we  do find $v_p=r_0+v_{p'}$, then 
$|\Omega'|=\dim(W)-\dim(W_0)|$,
otherwise $|\Omega'|=\dim(W)-\dim(W_0)|-1$.

Clearly, as the enumeration of the elements of $S_0$ requires at most $O(n^3)$ basic operations, the whole step requires at $O(\max(n,m)^3)$ basic operations.

The Lie subalgebra $\mathfrak{g}$ is isomorphic to
a direct sum of $2^{|\Omega'|}$ copies 
of $\mathfrak{so}(k)$.

\item 
If no, then find an $f$-hyperbolic basis $HB$ for $\langle v_p\mid p\in\mathcal{G}\rangle$
and for each basis vector $v$ determine $Q(v)$. 
Here $$HB=\{v_1,w_1,\dots, v_k,w_k,r_1,\dots, r_l\}$$
where $f(v_i,w_j)=\delta_{i,j}$ and $r_1,\dots, r_l$ in $\mathrm{Rad}(f)$.
This can be done in a modified Gram-Schmidt algorithm where in each step the value of $Q$ on the new vectors is computed. This will take at most $O(\max(n,m)^3)$ calculations.
\item 
Determine $Q(r_j)$ for $1\leq j\leq l$.

If not all these values are $0$, then 
$\mathfrak{g}$ is a direct sum of $2^{r-1}$ copies of 
$\mathfrak{su}(2^k)$. 

If all these values are $0$, then go to the next step.

This step requires $O(\ell)$ basic calculations.
\item For each hyperbolic pair $v_i,w_i$ of $HB$ determine the type of the restriction of $Q$ to  $\langle v_i,w_i\rangle$ 
by evaluating $Q$ at the three two vectors $v_i$ and $w_i$. 
Find  the type $T$ of $Q$ restricted to $\langle v_1,\dots,v_k\rangle$ as the product of the types
on the hyperbolic $2$-spaces.
This takes at most $O(n)$ computations.

The Lie subalgebra  $\mathfrak{g}$ is a direct sum of $2^r$ copies of 
$\mathfrak{so}(2^k)$ if the type $T$ is $+$ and of $2^r$ copies of 
$\mathfrak{sp}(2^{k-1})$ if the type $T$ is $-$. 

\end{enumerate}

As each of these eight steps requires running time at most $O(\max(n,m)^3)$, we find that the complexity of our algorithm is $O(\max(n,m)^3)$.

\end{document}